%% file: main.tex
\documentclass[12pt,transaction, draftclsnofoot,onecolumn]{IEEEtran}

\usepackage[utf8]{inputenc}
\usepackage{amsfonts}
\usepackage{amssymb}
\usepackage{hyperref}
\usepackage{graphicx}
\usepackage{subfigure}
\usepackage{enumerate}
\usepackage{amsmath}
\usepackage{color}
\usepackage{amsthm}
\usepackage{amsmath}
\usepackage{algorithm}
\usepackage[end]{algpseudocode}
 \usepackage{cite}
\usepackage{graphicx}
\usepackage[justification=centering]{caption}

\input{macros.tex}

\begin{document}
%
\title{Client Selection and Bandwidth Allocation in Wireless Federated Learning Networks: A Long-Term Perspective}
%
%
%

\author{Jie~Xu, Heqiang~Wang
\thanks{J. Xu and H. Wang are with the Department
of Electrical and Computer Engineering, University of Miami, FL, USA}
}

\maketitle

\begin{abstract}
This paper studies federated learning (FL) in a classic wireless network, where learning clients share a common wireless link to a coordinating server to perform federated model training using their local data. In such wireless federated learning networks (WFLNs), optimizing the learning performance depends crucially on how clients are selected and how bandwidth is allocated among the selected clients in every learning round, as both radio and client energy resources are limited. While existing works have made some attempts to allocate the limited wireless resources to optimize FL, they focus on the problem in individual learning rounds, overlooking an inherent yet critical feature of federated learning. This paper brings a new long-term perspective to resource allocation in WFLNs, realizing that learning rounds are not only temporally interdependent but also have varying significance towards the final learning outcome. To this end, we first design data-driven experiments to show that different temporal client selection patterns lead to considerably different learning performance. With the obtained insights, we formulate a stochastic optimization problem for joint client selection and bandwidth allocation under long-term client energy constraints, and develop a new algorithm that utilizes only currently available wireless channel information but can achieve long-term performance guarantee. Further experiments show that our algorithm results in the desired temporal client selection pattern, is adaptive to changing network environments and far outperforms benchmarks that ignore the long-term effect of FL.

\end{abstract}


%
\IEEEpeerreviewmaketitle

\section{Introduction}

Mobile devices nowadays generate a massive amount of data each day. This rich data has the potential to power a wide range of machine learning (ML)-based applications, such as learning the activities of smart phone users, predicting health events from wearable devices or adapting to pedestrian behavior in autonomous vehicles. Due to the growing storage and computational power of mobile devices as well as privacy concerns associated with uploading personal data, it is increasingly attractive to store and process data directly on each mobile device. The aim of ``federated learning'' (FL) \cite{konevcny2016federated} is to enable mobile devices to collaboratively learn a shared ML model with the coordination of a central server while keeping all the training data on device, thereby decoupling the ability to do ML from the need to upload/store the data in the cloud.

This paper focuses on FL in a classic wireless network setting where the clients, e.g., mobile devices, share a common wireless link to the server. We call this system a \textit{wireless federated learning network} (WFLN). The network operates for a number of learning rounds as follows: in each round, the clients download the current ML model from the server, improve it by learning from their local data, and then upload the individual model updates to the server via the wireless link; the server then aggregates the local updates to improve the shared model. Similar to a traditional throughput-oriented wireless network, the limited wireless network resources require the WFLN to determine in each round which clients access the wireless channel to upload the model updates and how much bandwidth is allocated to each client. However, due to the specific application in consideration, namely FL, the resource allocation objective and consequently the outcome can be very different from, e.g., throughput maximization.

Optimizing WFLNs faces unique challenges compared to optimizing either FL or the traditional wireless networks. On the one hand, the wireless network sets resource constraints on performing FL as the finite wireless bandwidth limits the number of clients that can be selected in each round, and the selection must be adaptive to the highly variable wireless channel conditions. On the other hand, FL is likely to change the way wireless networks should be optimized as model training is a complex \textit{long-term} process where decisions across rounds are interdependent and collectively decide the final training performance. Further, since mobile devices often have finite energy budgets due to, e.g., a finite battery, the number of rounds each individual mobile device can participate during the entire course of FL is also limited. An extremely crucial yet largely overlooked question is: does learning in different rounds contribute the same or differently to the final learning outcome and hence should the wireless resources be allocated discrepantly across rounds? Without a good understanding of its answer, conventional wireless network optimization approaches that treat each time slot independently and equally may lead to considerably suboptimal FL performance.

This paper aims to formalize this fundamental problem of client selection and bandwidth allocation in WFLNs and derive critical knowledge to enable the efficient operation of these networks. We study how resources (i.e., bandwidth and energy) should be allocated among clients in each learning round as well as \textit{across} rounds given finite client energy budgets in a volatile network environment. Our main contributions are summarized as follows.

(1) While existing works \cite{yang2020federated, zeng2019energy} have shown that including more clients in FL generally improves the learning performance, there is little understanding of how this improvement depends on the learning rounds. For a fixed total number of selected clients during the entire course of FL, should client selection be uniform across rounds or more biased toward the early/later FL rounds? Although analytical characterization seems extremely difficult, we show in two representative ML tasks, i.e., image classification and text generation, that selecting more clients in the later FL rounds not only achieves higher accuracy and lower training loss but also is more robust than selecting more clients in the early FL rounds. This finding, to our best knowledge, is the first that relates the temporal client selection pattern to the final FL performance.

(2) With the understanding of a desired temporal client selection pattern, we formulate a \textit{long-term} client selection and bandwidth allocation problem for a finite number of FL rounds under finite energy constraints of individual clients. Because wireless channel conditions vary over time but future conditions are unpredictable, we leverage the Lyapunov technique \cite{neely2010stochastic} to convert the long-term problem into a sequence of per-round problems via a virtual energy deficit queue for each client. A new online optimization algorithm called OCEAN is proposed, which in each FL round solves a finite number of convex optimization problems using only currently available wireless information, and hence the algorithm is practical and has low complexity.

(3) We prove that OCEAN achieves the FL performance of the desired client selection pattern within a bounded gap while approximately satisfying the energy constraints of the clients. Specifically, OCEAN demonstrates an $[O(1/V), O(\sqrt{V})]$ learning-energy tradeoff where $V$ is an algorithm parameter. In addition, we investigate the structure of the client selection and bandwidth allocation outcome. Our findings are two-fold: in each round, clients are selected according to a priority metric, which is the ratio of the client's current energy deficit queue length and its current wireless channel state. However, among the selected ones, more bandwidth is allocated to clients with a lower priority (i.e., worse channel and larger energy deficit queue length). This is in stark contrast to a traditional throughput-oriented wireless network where more bandwidth is allocated to clients with a better channel condition in order to maximize throughput.

\section{Related Work}
Since the proposal of FL \cite{konevcny2016federated, konevcny2016federated2}, a lot of research effort has been devoted to tackling various challenges in this new distributed machine learning framework, including developing new optimization and model aggregation algorithms \cite{karimireddy2019scaffold, li2018federated, haddadpour2019convergence}, handling non-i.i.d. and unbalanced datasets \cite{zhao2018federated,smith2017federated,corinzia2019variational}, and preserving model privacy \cite{bhowmick2018protection, geyer2017differentially, truex2019hybrid, bonawitz2017practical, nasr2018comprehensive} etc. Among these challenges, improving the communication efficiency of FL has been a key challenge due to the tension between uploading a large amount of data for model aggregation and the limited network resource to support this transmission. In this regard, a strand of literature focuses on modifying the FL algorithm itself to reduce the communication burden on the network, e.g., updating clients with significant training improvement \cite{chen2018lag}, compressing the gradient vectors via quantization \cite{lin2017deep}, or accelerating training using sparse or structured updates \cite{aji2017sparse,konevcny2016federated}. Hierarchical FL networks \cite{liu2019edge} have also been proposed where multiple edge servers perform partial model aggregation first, whose outputs are further aggregated by a cloud server. Recognizing the unique physical property of wireless transmission, \cite{yang2020federated,amiri2020federated} propose analog model aggregation over the air, provided that a very stringent synchronization is available.

As wireless networks are the envisioned main deployment scenario of FL, how to optimally allocate the limited bandwidth and energy resources for FL has also received much attention. Many existing works \cite{wang2019adaptive,tran2019federated,mo2020energy,zhan2020experience} study the inherent trade-off between local model update and global model aggregation, e.g., to adapt the frequency of global aggregation \cite{wang2019adaptive} or to optimize uplink transmission power/rate and the local update CPU frequency \cite{mo2020energy,zhan2020experience}. In all these works, all clients participate in every FL round. Although both empirical studies \cite{yang2020federated,zeng2019energy} and theoretical analysis \cite{stich2018local} show that including more clients improves the FL convergence speed, the limited bandwidth of wireless networks cannot support many clients to upload their local updates at the same time. For FL at scale, client scheduling policies, which select only a subset of clients in every round, are necessary. In \cite{yang2019scheduling}, the convergence performance of FL under three basic scheduling policies, namely random, round-robin and proportional fair, is analyzed. Different types of joint bandwidth allocation and client scheduling policies, e.g., \cite{zeng2019energy,chen2019joint,shi2019device,nishio2019client,yang2019energy,chen2020convergence}, have been proposed to either minimize the learning loss or the training time. However, their optimization problems are formulated by considering individual FL rounds separately or treating every FL round equally, and hence the same network resources are allocated across learning rounds. Our paper differs from these works in that we explicitly consider the varying significance of FL rounds and study a long-term bandwidth allocation and client selection problem under long-term energy constraints and with uncertain wireless channel information.

\section{Impact of Temporal Client Selection Pattern}
Existing works \cite{yang2020federated, zeng2019energy} have shown that the FL performance (in terms of training loss and prediction accuracy) can be improved by selecting more clients in each round. However, selecting more clients is not always possible if each client is subject to a long-term energy constraint due to, e.g., a finite battery: selecting more clients in early learning rounds  depletes the battery of the clients and hence fewer clients can be selected in later learning rounds. Hence, even with the same average number of selected clients, the temporal pattern can be considerably different, yet there is little understanding of how the temporal pattern affects the final FL outcome. In this section, we design two experiments to show that the temporal client pattern indeed has a considerable impact on the final FL performance.

\begin{figure}[htbp]
\centering
\begin{minipage}[t]{0.48\textwidth}
\centering
\includegraphics[width=8cm]{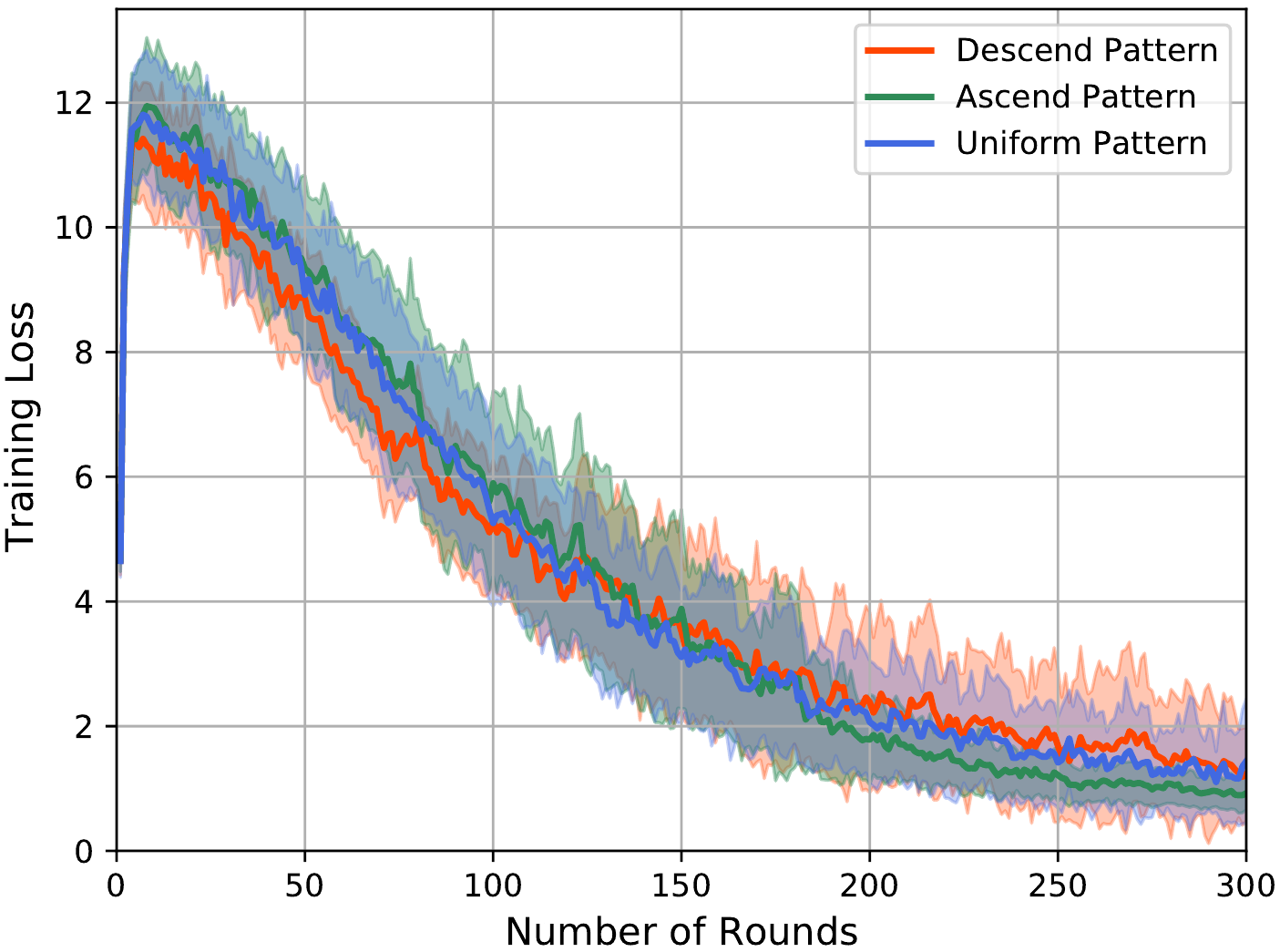}
\setlength{\abovecaptionskip}{0pt}
\caption{\label{fig1_1}Training Loss (MNIST Dataset)}
\end{minipage}
\begin{minipage}[t]{0.48\textwidth}
\centering
\includegraphics[width=8cm]{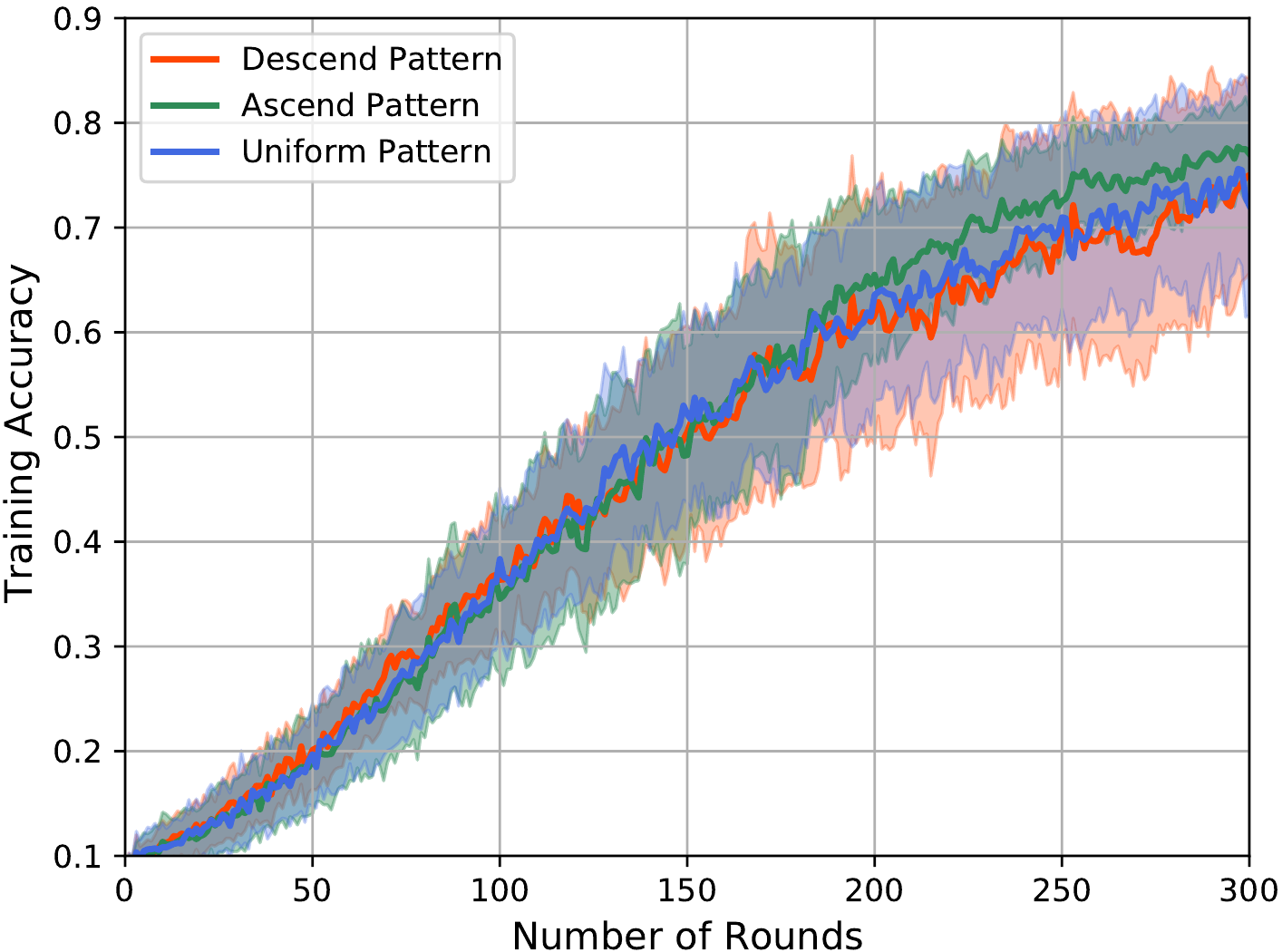}
\setlength{\abovecaptionskip}{0pt}
\caption{\label{fig1_2}Accuracy (MNIST Dataset)}
\end{minipage}
\end{figure}

\begin{figure}[htbp]
\centering
\begin{minipage}[t]{0.48\textwidth}
\centering
\includegraphics[width=8cm]{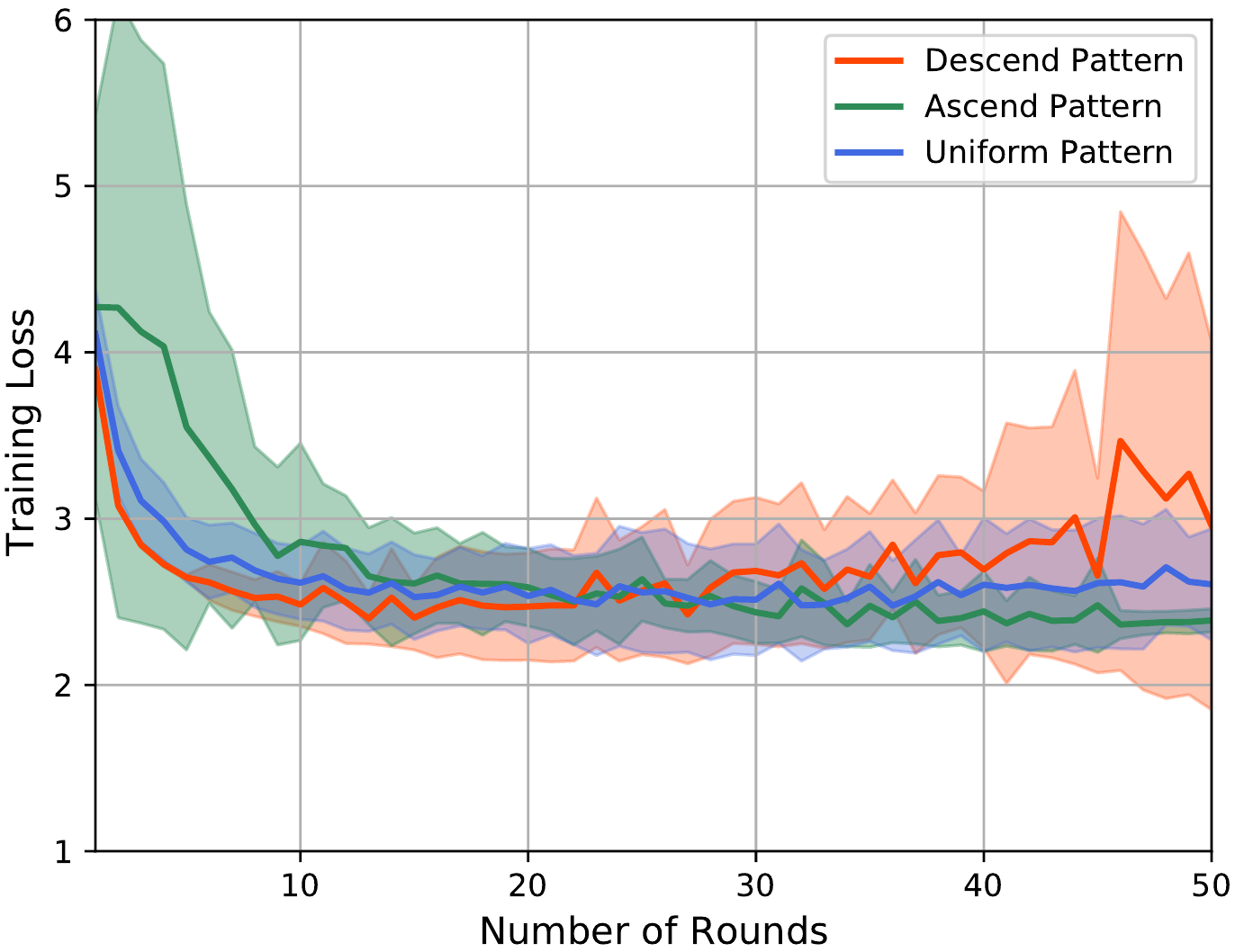}
\setlength{\abovecaptionskip}{0pt}
\caption{\label{fig2_2}Training Loss (Shakespeare Dataset)}
\end{minipage}
\begin{minipage}[t]{0.48\textwidth}
\centering
\includegraphics[width=8cm]{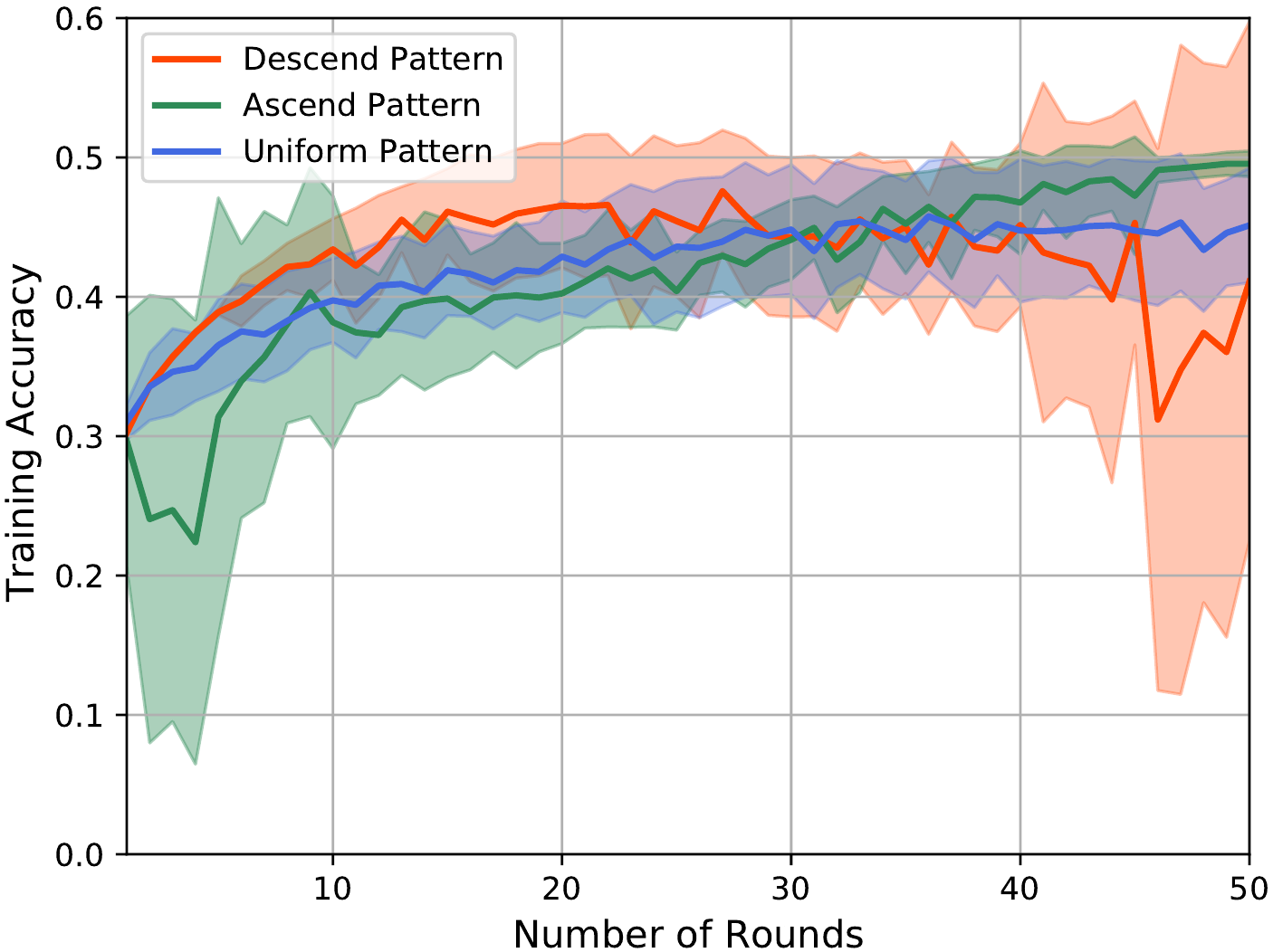}
\setlength{\abovecaptionskip}{0pt}
\caption{\label{fig2_1}Accuracy (Shakespeare Dataset)}
\end{minipage}
\end{figure}

\subsection{Image Classification on the MNIST Dataset}
Our first experiment is conducted using the TensorFlow Federated (TFF) framework \cite{tff2020} on the MNIST dataset for image classification. A deep neural network (DNN) classifier is trained on 10 clients (index from 1 through 10) using the FedAvg algorithm \cite{konevcny2016federated} over a total of 300 rounds. Three temporal selection patterns are investigated: \textbf{Uniform} -- in each round, 5 clients are randomly selected to upload their model parameters; \textbf{Ascend} -- the number of selected clients gradually increases from 1 to 10 over 300 rounds with an average number of 5 clients selected per round; \textbf{Descend} -- the number of selected clients gradually decreases from 10 to 1 over 300 rounds with an average number of 5 clients selected per round.

Figures \ref{fig1_1} and \ref{fig1_2} illustrate the training loss and the prediction accuracy, respectively, over 300 rounds for the three temporal patterns. Each curve is generated by averaging over 60 runs and the standard deviation of these curves are also shown in the figures. As can be seen, although the average number of selected clients is the same, different temporal patterns result in different training loss and prediction accuracy by the end of the 300 rounds. In particular, \textbf{Ascend} results in the best performance compared to \textbf{Uniform} and \textbf{Descend}. There is a good reason behind this result: early learning rounds are ``easy'' rounds where the learning performance is less sensitive to the number of selected clients. Hence, even if \textbf{Ascend} selects fewer clients in the early rounds, learning speed is minimally affected. However, the later learning rounds are the more ``difficult'' rounds, and to push accuracy even higher requires more clients to update the shared model using their data. In fact, not only \textbf{Ascend} wins in training loss and accuracy, but it is also much more robust as the standard deviation is much smaller. This is again because more clients participate in model updating towards the end of learning, which can smooth out abrupt changes in the learned model of individual clients.

\subsection{Text Generation on the Shakespeare Dataset}
To verify that the above findings are generalizable, we conduct a similar experiment on a text generation task. We utilize the decentralized text generation dataset based on \textit{The Complete Works of Shakespeare} provided in the TensorFlow Federated tutorial \cite{tff2020}. A Recurrent Neural Network  with eager execution is pre-trained on the text from Charles Dickens' \textit{A Tale of Two Cities} and \textit{A Christmas Carol} as the initial model, and FL is used to fine-tune this model for the Shakespeare dataset. As can be seen in Figures \ref{fig2_2} and \ref{fig2_1}, although the task and the dataset are very different, similar observations can be made as before: the \textbf{Ascend} selection pattern significantly outperforms \textbf{Descend} and \textbf{Uniform} in terms of training loss, accuracy and robustness.

We note that the exact optimal temporal selection pattern seems impossible to analytically characterize, which would also change across different learning tasks, models, datasets and algorithms. However, as the general ascending trend leads to considerable performance improvement, it offers valuable guidance for designing client selection schemes across rounds.

\section{Wireless Federated Learning Network Model}
With the insights obtained in Section III, we now move on to optimize the WFLN. Consider a WFLN with one server and $K$ clients, indexed by the set $\mathcal{K} = \{1,...,K\}$. Each participating client $k \in \mathcal{K}$ has a local dataset $\mathcal{D}_k$. In the supervised learning case, $\mathcal{D}_k$ defines the collection of data samples given as a set of input-output pairs $\{x_i, y_i\}_{i=1}^{D_k}$, where $x_i \in \mathbb{R}^d$ is a $d$-dimensional input feature vector, and $y_i \in \mathbb{R}$ is the ground-truth output label. This data can be generated through the usage of the client via mobile applications and can be employed for various ML tasks, e.g., user activity prediction or health event prediction.

FL iterates between two steps: 1) the server updates the global model by aggregating local models transmitted over a multi-access channel by the clients; 2) the clients update their local models using the global model broadcasted by the server. We call each iteration a \textit{learning round}. WFLN has to decide in each round which clients upload their local model updates depending on their wireless channel condition and remaining battery to maximize the learning performance. We use $a^t_k \in \{1, 0\}$ to denote whether or not client $k$ is selected in round $t$, and $\boldsymbol{a}^t = (a^t_1, ..., a^t_K)$ collects the overall client selection decisions.

\subsection{Client Energy Consumption}
For a selected client $k$ in round $t$ (i.e., $a^t_k = 1$), it incurs energy consumption due to uploading the local updates to the edge server via the wireless channel. We consider a specific wireless multi-access scheme, i.e., orthogonal frequency-division multiple access (OFDMA) for local model uploading with a total bandwidth $B$. Let $b^t_k \in [0,1]$ be the bandwidth allocation ratio for client $k$ in round $t$, and hence its allocated bandwidth is $b^t_k B$. Let $\b^t = (b^t_1, ..., b^t_K)$. Bandwidth allocation must satisfy $\sum_{k\in\mathcal{K}} b^t_k = 1, \forall t$. Clearly, if $a^t_k = 0$, namely client $k$ is not selected in round $t$, then it is the best not to allocate any bandwidth to this client, i.e., $b^t_k = 0$. On the other hand, if $a^t_k = 1$, then we require that at least a minimum bandwidth $b_{min}$ is allocated to client $k$, i.e., $b^t_k \geq b_{min}$. This is because practical systems cannot assign an arbitrarily small bandwidth to an individual client. In addition, a close-to-zero bandwidth allocation will require an extremely high transmit power and hence result in an extremely high energy consumption to achieve a target transmission rate. To make the problem feasible, we assume $b_{min} \leq 1/K$.

Let $p_k^t$ denote the transmission power (in Watt/Hz) of client $k$ in round $t$. The achievable rate (in bit/s), denoted by $r_k^t$, can be written according to the Shannon's formula as
\begin{align}\label{tx_rate}
r^t_k = b^t_k B \log_2\left(1 + \frac{p^t_k (h^t_k)^2}{N_0}\right)
\end{align}
where $N_0$ is the variance of the complex white Gaussian channel noise and $h^t_k$ is the channel state of client $k$ in round $t$. Let $L$ denote the data size of the adopted machine learning model (in bit), then the time needed to upload the local model update to the edge server is $\tau_k = L/r^t_k$. For a target  upload time deadline $\bar{\tau}$, the required transmission power can be derived using \eqref{tx_rate} and hence the transmission energy consumption of client $k$ is
\begin{align}
E(a^t_k, b^t_k|h^t_k) = \frac{\bar{\tau}N_0 B b_k^t}{(h^t_k)^2}  \left(2^{\frac{L}{\bar{\tau}B b^t_k}}-1\right) a^t_k
\end{align}

\subsection{System Learning Performance}
Existing works and our empirical study show that in order to accelerate learning, it is desirable for the WFLN to select as as many clients as possible in each round. However, client selection is constrained by finite radio and battery resources. Therefore, the WFLN must judiciously select clients to perform federated learning in each round, without quickly draining clients' battery and causing insufficient model updates in later communication rounds. To this end, we introduce the following metric to describe the FL performance in round $t$:
\begin{align}\label{learn_per}
U^t(\a^t) = \eta^t \sum_{k=1}^K a^t_k
\end{align}
where $\eta^t$ is a temporal weight to capture the varying significance of selecting more clients in different learning rounds. As suggested by Section III, an increasing sequence of $\eta^t$ often results in better FL performance as more clients are likely to be selected in later rounds of learning.

We note, however, although the above metric will facilitate our subsequent resource allocation, it does not exactly characterize the FL speed or accuracy, which is extremely difficult, if not impossible, to model due to the complex and non-convex nature of many ML algorithms.




\subsection{Problem Formulation}
As we emphasize the long-term performance and final outcome of FL, the goal is to maximize the weighted sum of selected clients defined in \eqref{learn_per} for a total number of $T$ learning rounds while satisfying the long-term energy budget constraints of individual clients, through joint client selection $\a^t$ and bandwidth allocation $\b^t$ in every round $t = 0, ..., T-1$. Although the performance metric defined in \eqref{learn_per} is artificial, we will relate it to the actual FL performance (i.e., training loss and accuracy) in experiments. Formally, the problem that we aim to solve is
\begin{align}
\textbf{P1}&~~~\max_{\a^0, \b^0, ..., \a^{T-1}, \b^{T-1}} \sum_{t=0}^{T-1} U^t(\a^t)\\
\text{s.t.}&~~~\sum_{t=0}^{T-1} E(a^t_k, b^t_k|h^t_k) \leq H_k, \forall k \label{con:battery}\\
&~~~b_{min} \leq b^t_k \leq 1, \forall k, \forall t,~~~\sum_{k=1}^K b^t_k = 1, \forall t \label{con:band}\\
&~~~a^t_k \in \{0, 1\}, \forall k, \forall t \label{con:selection}
\end{align}
Constraint \eqref{con:battery} requires that the total energy consumption over the $T$ rounds for each client $k$ does not exceed an energy budget $H_k$ (e.g., battery capacity or energy limit set by the client). Constraint \eqref{con:band} is the feasibility condition on the bandwidth allocation. Constraint \eqref{con:selection} is the feasibility condition on the client selection.

So far we have formulated a long-term optimization problem for client selection and bandwidth allocation in WFLNs. However, several challenges impede the derivation of the optimal solution to \textbf{P1}. The first is the lack of future information: optimally solving \textbf{P1} requires complete offline information (i.e., channel conditions) over the entire FL period (i.e., $T$ learning rounds) that is very difficult to accurately predict in advance. Furthermore, \textbf{P1} belongs to mixed-integer nonlinear programming and is difficult to solve, even if the long-term future information is accurately known a priori. Thus, these challenges demand an online approach that can efficiently make joint client selection and bandwidth allocation decisions without foreseeing the far future.

\subsection{Offline Benchmark: $R$-Round Lookahead Algorithm}
Before we propose the online algorithm, we first introduce an offline algorithm with $R$-round lookahead information (i.e., the channel information in the next $R$ learning rounds are \textit{assumed} to be known) as a benchmark. Specifically, we divide the entire FL period into $M \geq 1$ frames, each having $R \geq 1$ learning rounds such that $T = MR$, and present the following problem formulation:
\begin{align}
\textbf{P2}:&~~~\max_{\a^0, \b^0, ..., \a^{T-1}, \b^{T-1}} \sum_{t = mR}^{(m+1)R - 1} U^t(\a^t)\\
\text{s.t.}&~~~\sum_{t=mR}^{(m+1)R -1} E(a^t_k, b^t_k|h^t_k) \leq H_k/M, \forall k\\
&~~~\text{Constraints \eqref{con:band}, \eqref{con:selection}} \nonumber
\end{align}
Essentially, \textbf{P2} defines a family of offline algorithms parameterized by the lookahead window size $R$.  Clearly, there exists at least one sequence of joint client selection and bandwidth allocation decisions that satisfies all constraints of \textbf{P2} (e.g., no client is selected in any round in each frame). We denote the optimal learning performance for the $m$-th frame by $U^*_m$, for $m = 0, ..., M-1$, considering all the decisions that satisfy the constraints and have perfect information over the frame. Thus, the optimal long-term learning performance achieved by the oracle's optimal $R$-round lookahead algorithm is given by $\sum_{m=0}^{M-1} U^*_m$.

We note that because of the assumed lookahead information, the $R$-round lookahead algorithms are impractical (unless $R = 1$). The purpose of introducing these algorithms is only to use them as a benchmark for our practical online algorithm to be proposed in the next section.

\section{Online Client Selection and Bandwidth Allocation}
In this section, we develop the \underline{O}nline \underline{C}lient s\underline{E}lection and b\underline{A}ndwidth allocatio\underline{N} algorithm, called OCEAN, and then characterize its structural properties. We also prove that it is efficient compared to the optimal offline algorithm with $R$-round lookahead information.

\subsection{The OCEAN Algorithm}
A major challenge of directly solving \textbf{P1} is that the long-term energy constraint of the clients couples the client selection and bandwidth allocation decisions across different learning rounds: selecting more clients in the current round reduces the bandwidth allocated to each individual client, thereby increasing the energy consumption of these clients; furthermore, more energy consumption in the current round potentially reduces the energy budget available for future FL rounds, and yet the decisions have to be made without foreseeing the future. To address this challenge, we leverage the Lyapunov technique and construct a virtual energy deficit queue $q_k(t)$ for each client $k$ to guide the client selection and bandwidth allocation decisions to follow the long-term energy constraint. The virtual energy queue of client $k$ starts with $q_k(0) = 0, \forall k$, and is updated at the end of each round $t$ as follows
\begin{align}\label{queue}
q_k(t+1) = [E(a^t_k, b^t_k|h^t_k) - H_k/T + q_k(t)]^+
\end{align}
where $[\cdot]^+ = \max\{\cdot, 0\}$. Hence, $q_k(t)$ is the queue length indicating the deviation of the current energy consumption of client $k$ from its long-term energy constraint $H_k$. Let $\q(t) = (q_1(t), q_2(t), ..., q_K(t))$ collect the energy deficit queues for all clients.

\begin{algorithm}[t]
	\caption{OCEAN}
	\begin{algorithmic}[1]
		\State \textbf{Input}: $q_k(0) = 0, \forall k$ and $R$
        \For {$t = 1, 2, ..., T$}
            \If {$t = mR, \forall m = 1, ..., M-1$}
            \State $q_k(t) \leftarrow 0, \forall k$ and $V \leftarrow V_m$
            \EndIf
            \State Observe the current channel state $h^t_k, \forall k$
            \State Solve \textbf{P3}
            \State Update energy queue according to \eqref{queue}
        \EndFor
	\end{algorithmic}\label{alg:OCEAN}
\end{algorithm}

We now present OCEAN in Algorithm \ref{alg:OCEAN}. OCEAN is purely online and requires only the currently available channel state information as inputs (i.e. $h_k(t), \forall k$). We use $V_0, V_1, ..., V_{M-1}$ to denote a sequence of positive control parameters to dynamically adjust the tradeoff between maximizing the number of selected clients and minimizing energy consumption over the $M$ frames, each having $R$ communication rounds. The importance of the control parameters will be revisited in Section V.C. In every round $t$, we aim to solve the following per-round problem:
\begin{align}
\textbf{P3}&~~~\max_{\a^t, \b^t} V\cdot U^t(\a^t) - \sum_{k=1}^K q_k(t) E(a^t_k, b^t_k|h^t_k)\\
\text{s.t.} &~~~ \eqref{con:band},\eqref{con:selection}
\end{align}
By considering the additional term $ \sum_{k=1}^K q_k(t) E(a^t_k, b^t_k|h^t_k)$, the system takes into account the energy deficit of the clients during the current round's client selection and bandwidth allocation. As a consequence, when $q_k(t)$ is larger, minimizing the energy deficit is more critical. Thus, our algorithm works following the philosophy of ``if violate the energy constraint, then use less energy'', and the energy deficit queue maintained without foreseeing the future guides the system towards meeting the energy constraints of the clients. OCEAN decomposes the long-term optimization problem into a series of per-round problems \textbf{P3}. For a more rigorous derivation of this decomposition, please refer to the proof of Theorem 1. Now, to complete OCEAN, it remains to solve \textbf{P3}, which however is still very difficult.

\subsection{Solving the Per-Round Problem}
The per-round problem \textbf{P3} is a difficult mixed-integer problem. To see more clearly how the objective function depends on $\a^t$ and $\b^t$, we write out and rearrange it as follows
\begin{align}
\max_{\a^t, \b^t}~~ \sum_{k=1}^K\left[V\eta^t - q_k(t-1)\cdot \frac{\bar{\tau}N_0 B b_k^t}{(h^t_k)^2}  \left(2^{\frac{L}{\bar{\tau}B b^t_k}}-1\right)\right] a^t_k
\end{align}
Notice that $a^t_k$ is a binary integer variable and $b^t_k$ is a continuous variable in $[b_{\text{min}}, 1]$. In general, mixed-integer problems are difficult to solve and often there is no polynomial-time optimal algorithm. Fortunately, our problem \textbf{P3} exhibits a special structure and we are able to exploit this structure to develop an algorithm that returns the optimal solution by solving at most $K$ convex optimization problems. To simplify the notations, we drop the index $t$ in this subsection.

Our algorithm to solve \textbf{P3}, called OCEAN-P, incrementally adds clients into the selection set $S$ based on a metric $\rho_k \triangleq \frac{q_k(t)}{(h^t_k)^2}$, which we call the \textit{selection priority} (the lower value, the higher priority). Initially, all clients with $\rho_k = 0$ (which also means $q_k(t) = 0$ as $(h^t_k)^2$ is always positive) are added into $S$. We denote this initial set by $S^0$. Then, clients with $\rho_k > 0$ are added into $S$ one by one in the ascending order of $\rho_k$, and for each possible selection set, the corresponding bandwidth allocation is computed by solving the following optimization problem
\begin{align}
\textbf{P4}~~~\max_{\{b^t_k\}_{k \in S - S^0}}&~~ \sum_{k\in S -S^0} \left(V\eta^t - \rho_k N_0  \tilde{\tau} B b_k\left(2^\frac{L}{\tilde{\tau} B b_k}-1\right)\right)\\
\text{s.t.}&~~\sum_{k \in S-S^0} b^t_k = 1 - |S^0|\cdot b_{min}\\
&~~b_k \geq b_{min}, \forall k \in S-S^0
\end{align}
Let $\b^*(S)$ be the optimal bandwidth allocation for a given selection set $S$, and $W^*(S)$ be the optimal value. Clearly, for the initial set $S^0$, $W^*(S^0) = \eta^t|S^0|$ as $\rho_k = 0, \forall k \in S^0$. The number of selection sets that can possibly emerge following the above set expanding rule, which are collected in $\mathcal{S}$, is at most $K$. Finally, the implemented optimal selection is $S^* = \arg\max_{S\in\mathcal{S}} W^*(S)$ and the implemented optimal bandwidth allocation is $\b^* = \b^*(S^*)$.

Because there are $K$ clients and in every iteration, one more client is added into $S$, we ensure that the algorithm only needs to solve at most $K$ optimization problems \textbf{P4} to return the optimal solution. In fact, we can reduce the number of times for solving \textbf{P4} by adding a termination condition: if for some $S$, its optimal bandwidth allocation $\b^*(S)$ results in $\eta^t - \rho_k N_0  \tilde{\tau} B b_k\left(2^\frac{L}{\tilde{\tau} B b^t_k}-1\right) < 0$ for the last added client $k$, then the algorithm stops adding more clients into the selection set. This termination condition can significantly reduce the number of convex optimization problems to be solved when $K$ is large. The pesudocode of OCEAN-P is given in Algorithm \ref{alg:CSBA}. Next, we first prove that \textbf{P3} is a convex optimization problem and then prove the optimality of OCEAN-P.

\begin{algorithm}[t]
	\caption{OCEAN-P}
	\begin{algorithmic}[1]
        \State \textbf{Input}: $q_k(0) = 0, \forall k$
        \State Rank the clients according to $\rho$. Hence we have $\rho_1\leq \rho_2 \leq ... \leq \rho_K$
        \State Set $S^0 = \{k: \rho_k = 0\}$, $S = S^0$, and $\mathcal{S} = \{S^0\}$.
        \For {$k = |S^0|+1, ..., K$}
            \State Update $S = S\cup\{k\}$
            \State Solve \textbf{P4} and obtain $\b^*(S)$ and $W^*(S)$
            \If {$V\eta^t - \rho_k N_0  \tilde{\tau} B b^t_k\left(2^\frac{L}{\tilde{\tau} B b_k}-1\right) < 0$}
                \State Stop iteration
            \Else {}
                \State Add $S$ to $\mathcal{S}$, i.e. $\mathcal{S} = \mathcal{S} \cup \{S\}$
            \EndIf
        \EndFor
        \State Find $S^* = \arg\max_{S\in\mathcal{S}} W^*(S)$
        \State Return $\a^*$ where $a^*_k = \mathbf{1}\{k \in S^*\}, \forall k$ and $\b^* = \b^*(S^*)$
	\end{algorithmic}\label{alg:CSBA}
\end{algorithm}

\begin{lemma}
The function $f(x) = x(2^\frac{\beta}{x} -1)$ where $\beta > 0$ is decreasing and convex in $x \in (0, \infty)$ and is increasing and concave in $x \in (-\infty, 0)$.
\end{lemma}
\begin{proof}
See Appendix \ref{proofLemma1}.
\end{proof}

Lemma 1 readily proves that \textbf{P4} is a convex optimization problem as $b^t_k \in [b_{min}, 1)$. Convex optimization problems are extensively studied in the literature and many efficient algorithms \cite{boyd2004convex} and mature software tools (such as CVX \cite{grant2014cvx} and SciPy \cite{virtanen2020scipy}) exist. Next, we prove that our algorithm returns the optimal solution by solving \textbf{P4} at most $K$ times.

\begin{theorem}
OCEAN-P returns the optimal solution to the per-round problem \textbf{P3} by solving at most $K$ convex optimization problems.
\end{theorem}
\begin{proof}
See Appendix \ref{proofTheorem1}.
\end{proof}

\subsection{Structural Results and Performance Analysis}
In this subsection, we first investigate the structure of the optimal solution produced by OCEAN-P in every round, and then characterize the performance of OCEAN.

In Theorem 1, we have already proven a thresholding result on the client selection, namely only clients whose selection priority $\rho_k$ is below a threshold are selected to participate in a FL round. Proposition 1 characterizes how bandwidth is allocated among the selected clients and their incurred energy consumption.

\begin{proposition}
In any learning round $t$, the allocated bandwidth $b^{t,*}_k$ of a selected client $k$ and its weighted energy consumption $q_k(t)E^t_k(b^{t,*}_k)$ are non-decreasing with $\rho^t_k$. 
\end{proposition}
\begin{proof}
See Appendix \ref{proofProposition1}.
\end{proof}

Theorem 1 and Proposition 1 together show that a client with a smaller energy deficit $q_k(t)$ and a better channel condition $(h_k^t)^2$ (and hence a smaller $\rho_k^t$) is more likely to be selected to participate in the current FL round; however, among the selected clients, a client with a smaller $\rho^t_k$ is allocated with less bandwidth. This is because although allocating more bandwidth to client $k$ with a smaller $\rho_k^t$ reduces the energy consumption and deficit of this client, it reduces the bandwidth that can be allocated to clients with larger $\rho$, which leads to even higher increased energy consumption and deficit of those clients. Moreover, in the optimal solution, the overall effect of energy deficit and consumption, namely $q_k(t-1) E^t_k(b^*_k)$, is still increasing in $\rho^t_k$.

With the optimality of OCEAN-P, we then prove the performance guarantee of OCEAN.

\begin{theorem}
For any $R \in \mathbb{Z}^+$ and $M\in\mathbb{Z}^+$ such that $T = MR$, when comparing OCEAN with the $R$-round lookahead algorithm, the following statements hold:

(a) The energy constraint of every client $k$ is approximately satisfied with a bounded deviation:
\begin{align}
\sum_{t=0}^T E_k(a^t_k, b^t_k|h^t_k) \leq H_k + \sum_{m=0}^{M-1} \sqrt{\frac{2(V_m\eta^t K + C_1)}{R}}, \forall k
\end{align}
where $C_1 \triangleq K (E^\text{max} - H^{\text{min}}/T)^2/2$.

(b) The federated learning performance satisfies:
\begin{align}
\sum_{t=0}^{T-1}U(\a^t) \geq \sum_{m=0}^{M-1} U^*_m - C_2 \sum_{m=0}^{M-1}\frac{1}{V_m}
\end{align}
where $C_2 \triangleq C_1 R + \frac{R(R-1)K}{2}(E^\text{max})^2$ and $U^*_m$ is the optimal value achieved by the $R$-round lookahead algorithm in frame $m$.
\end{theorem}

\begin{proof}
See Appendix \ref{proofTheorem2}.
\end{proof}

Theorem 2 shows that, given a fixed value of $R$ and $M$, OCEAN is $O(1/V)$-optimal with respect to the FL performance against the optimal $R$-lookahead policy, while the energy consumption is guaranteed to be approximately satisfied with a bounded factor $O(\sqrt{V})$. Thus, OCEAN demonstrates an $[O(1/V), O(\sqrt{V})]$ learning-energy tradeoff. Note that when $R = T$, the $T$-lookahead benchmark has complete future information of the entire $T$ rounds. Even in this case, the $[O(1/V), O(\sqrt{V})]$ tradeoff still holds.

\section{Simulation Results}
In this section, we simulate a WFLN to evaluate the performance of OCEAN.

\textbf{Federated Dataset}. To simulate FL, we leverage the TensorFlow Federated (TFF) framework and the MNIST dataset for hand-written digit classification. Each client's local dataset is keyed by the original writer of the digits. Since each writer has a unique style, this dataset exhibits the kind of non-i.i.d. behavior expected of federated datasets. We use the first 10 clients in the MNIST dataset to conduct our simulation with each client having 100 training data samples. Since the hand-written digit classification is a relatively easy image classification task, we follow TFF's tutorial to construct a simple three-layer neural network with the first layer being input, the second containing 10 neurons and the third performing the softmax operation. This neural network's model size is $L = 3.4 \times 10^5$ bits. FedAvg \cite{konevcny2016federated} is used as the learning algorithm.

\textbf{Wireless Network}. To simulate the wireless network, we consider an OFDMA system where the total bandwidth $B = 10$ MHz. Each client's wireless channel gain is modelled as independent free-space fading with average path loss 36dB. The variance of the complex white Gaussian channel noise is set as $N_0 = 10^{-12}$ W. To ensure timely model update, we set the target uploading time in each round to be $\bar{\tau} = 300$ ms. The minimal bandwidth $b_{min}$ is set as $2 \times 10^5$ Hz. For each client $k$, the energy budget is set as $H_k = 0.15$ J. The network runs for $T = 300$ rounds.

\subsection{Benchmarks}
We compare the performance of OCEAN with the following three benchmark algorithms.
\begin{itemize}
  \item \textbf{Select-All}: All 10 clients are selected in every learning round. Bandwidth is allocated to minimize the total energy consumption while satisfying the upload deadline requirement. 
  \item \textbf{Static Myopic Optimal (SMO)}: In every learning round, SMO uses only currently available information independently across rounds (which is equivalent to the 1-Round Lookahead algorithm) to solve
\begin{align}
\max_{\a^t, \b^t} &~~~\sum_{k} a^t_k\\
\text{s.t.}&~~~E(a^t_k, b^t_k|h^t_k) \leq H_k/T, \forall k\\
&~~~\text{Constraints \eqref{con:band}, \eqref{con:selection}} \nonumber
\end{align}
  This problem is easy to solve: for each client $k$, first compute the required bandwidth $b^\dagger_k \geq b_{min}$ so that using $H_k/T$ energy can meet the upload time target $\bar{\tau}$; then rank $b^\dagger_k$ in the ascending order and select clients until the total required bandwidth exceeds $B$. SMO mimics existing approaches (e.g., \cite{zeng2019energy}) that solves bandwidth allocation and client selection independently across learning rounds.
  \item \textbf{Adaptive Myopic Optimal (AMO)}: SMO has a clear deficiency which can result in energy under-utilization: when a client is not selected in a round, its energy $H_k/T$ is wasted and will not be used in future rounds. To address this issue, we also consider a modified version of SMO, which recycles previously unused energy budget for future rounds. In particular, the energy budget for client $k$ in round $t$ is modified to $(H_k - \sum_{\tau = 0}^{t-1} E^t_k)/(T-t)$.
  \end{itemize}

For OCEAN, we let $R = T$ and hence the sequence $V_1, ..., V_M$ becomes a single scalar $V$. Moreover, we implement three variants using different temporal importance sequences $\eta_t$: Ascending (OCEAN-a); Descending (OCEAN-d); and Uniform (OCEAN-u).

\subsection{Performance Comparison}
Figure \ref{fig3} shows the number of selected clients in every round for different approaches, which is obtained by averaging over 10 runs. As the name suggests, Select-All selects all 10 clients in every round, resulting in the ideal optimal client selection for FL. SMO selects much fewer clients due to the hard energy budget allocation in every round. Many clients do not get to upload their local model updates due to the bad channel state that they are experiencing. AMO starts with selecting few clients due to the same reason as SMO. However, as time goes on, energy budget not used in the previous rounds accumulates. This allows the client to transmit at the desired rate using a higher transmission power in later rounds, especially in those towards the very end, thereby countering the effects of bad channel states. As a (fortunate) by-product, AMO also achieves an ascending pattern of client selection. Our proposed algorithm, OCEAN-a, is able to select many more clients than SMO because it uses energy as needed without imposing a hard per-round energy constraint. Compared to AMO, it is able to fine-tune the temporal pattern of client selection by using different sequences of temporal weights $\eta$. As can be seen in Figure \ref{fig13}, OCEAN-a results in an increasing number of selected clients, OCEAN-d results in a decreasing number of selected clients, while OCEAN-u keeps the number of selected clients almost the same across rounds.

\begin{figure}[htbp]
\centering
\begin{minipage}[t]{0.48\textwidth}
\centering
\includegraphics[width=8cm]{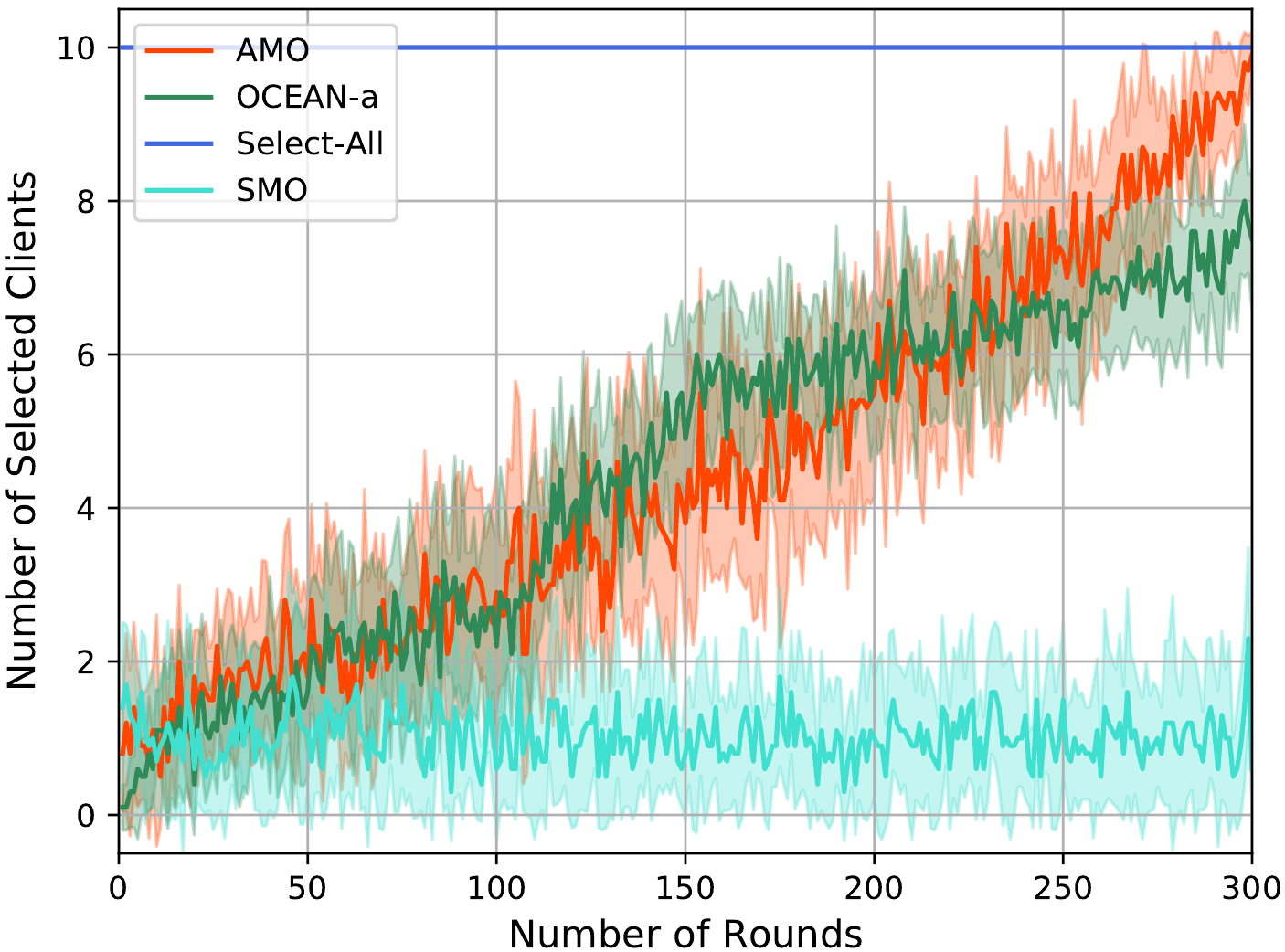}
\setlength{\abovecaptionskip}{0pt}
\caption{\label{fig3} Temporal Client Selection Patterns of OCEAN and Benchmarks}
\end{minipage}
\begin{minipage}[t]{0.48\textwidth}
\centering
\includegraphics[width=8cm]{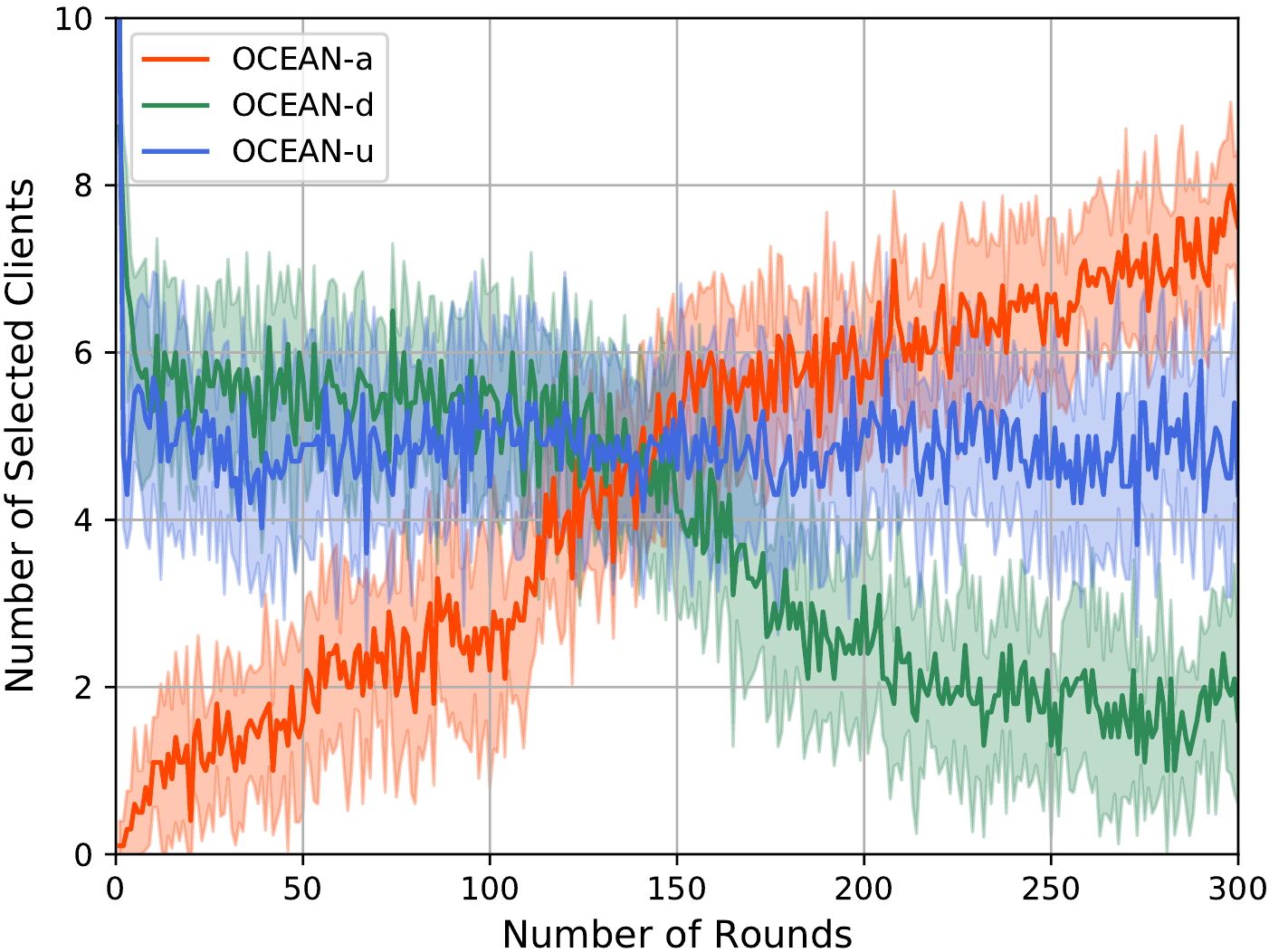}
\setlength{\abovecaptionskip}{0pt}
\caption{\label{fig13}Temporal Client Selection Patterns of OCEAN Variants}
\end{minipage}
\end{figure}

Figure \ref{fig4} shows the actual energy consumption of individual clients by the end of 300 learning rounds for different approaches in a particular run. Because Select-All completely ignores the energy budgets of the clients, it results in a very large energy consumption, far exceeding the energy budgets. On the other hand, SMO does not fully utilize the client's energy budget because in many learning rounds the client is not selected. Both AMO and OCEAN-a incur a total energy consumption close to the given energy budget (i.e. 0.15) for individual clients.

\begin{figure}[htbp]
\centering
\includegraphics[width=8cm]{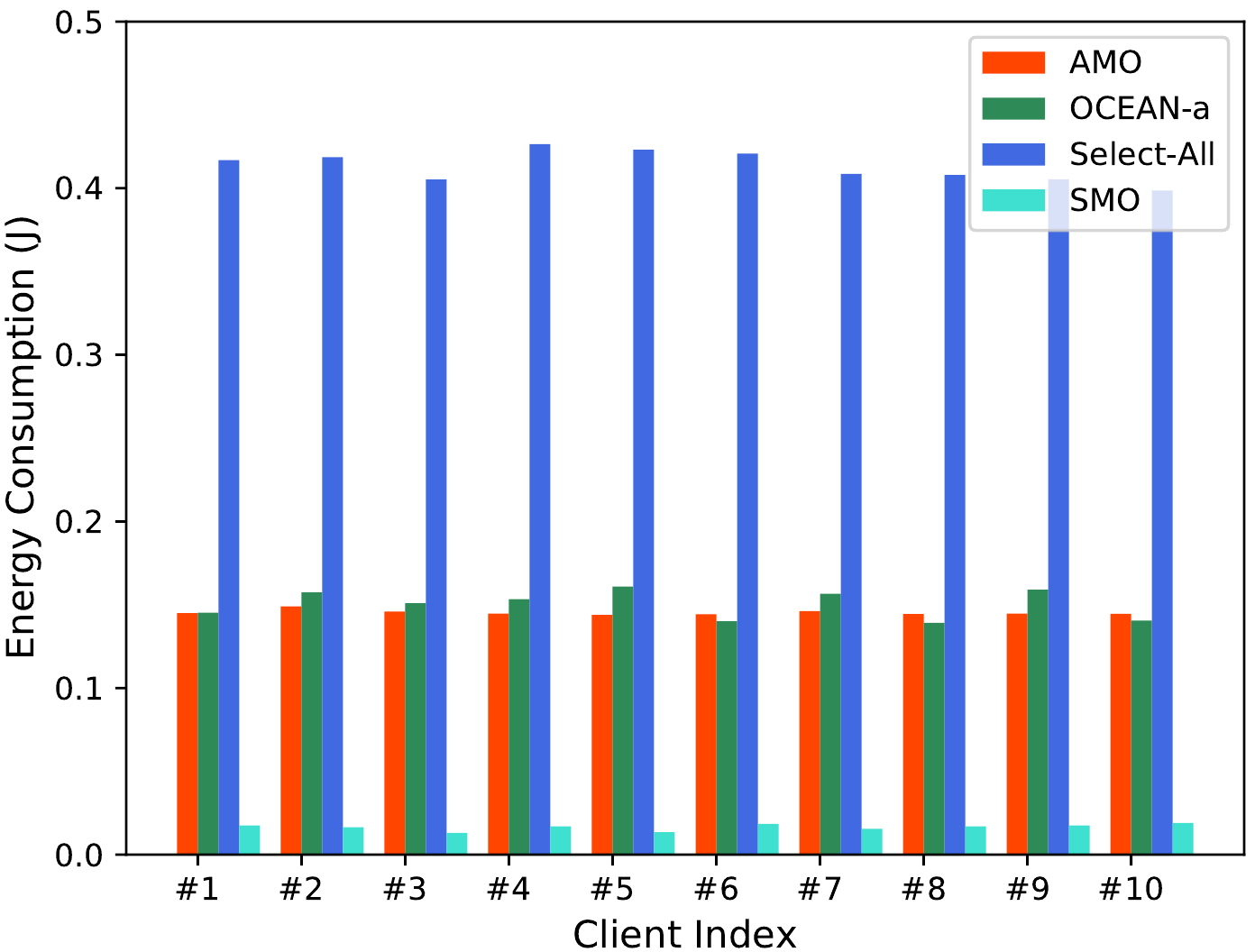}
\caption{\label{fig4}Per-Client Energy Consumption Comparison}
\end{figure}

As our ultimate goal is to improve the FL performance, we show the training loss and accuracy for different approaches in Figures \ref{fig5_1} and \ref{fig5_2}. Select-All, as expected, results in the best FL performance, with the smallest training loss, the highest accuracy and the fastest convergence among all approaches. Due to the insufficient selection of clients in the course of learning, SMO's learning performance is considerably inferior to all other approaches. Thanks to the fortunate by-product of AMO, AMO's FL performance is comparable to OCEAN-a in this specific setting, which is close to the ideal case Select-All. However, we will show in the next set of experiments that AMO's ``luck'' does not extend to other more complex network environments.

\begin{figure}[htbp]
\centering
\begin{minipage}[t]{0.48\textwidth}
\centering
\includegraphics[width=8cm]{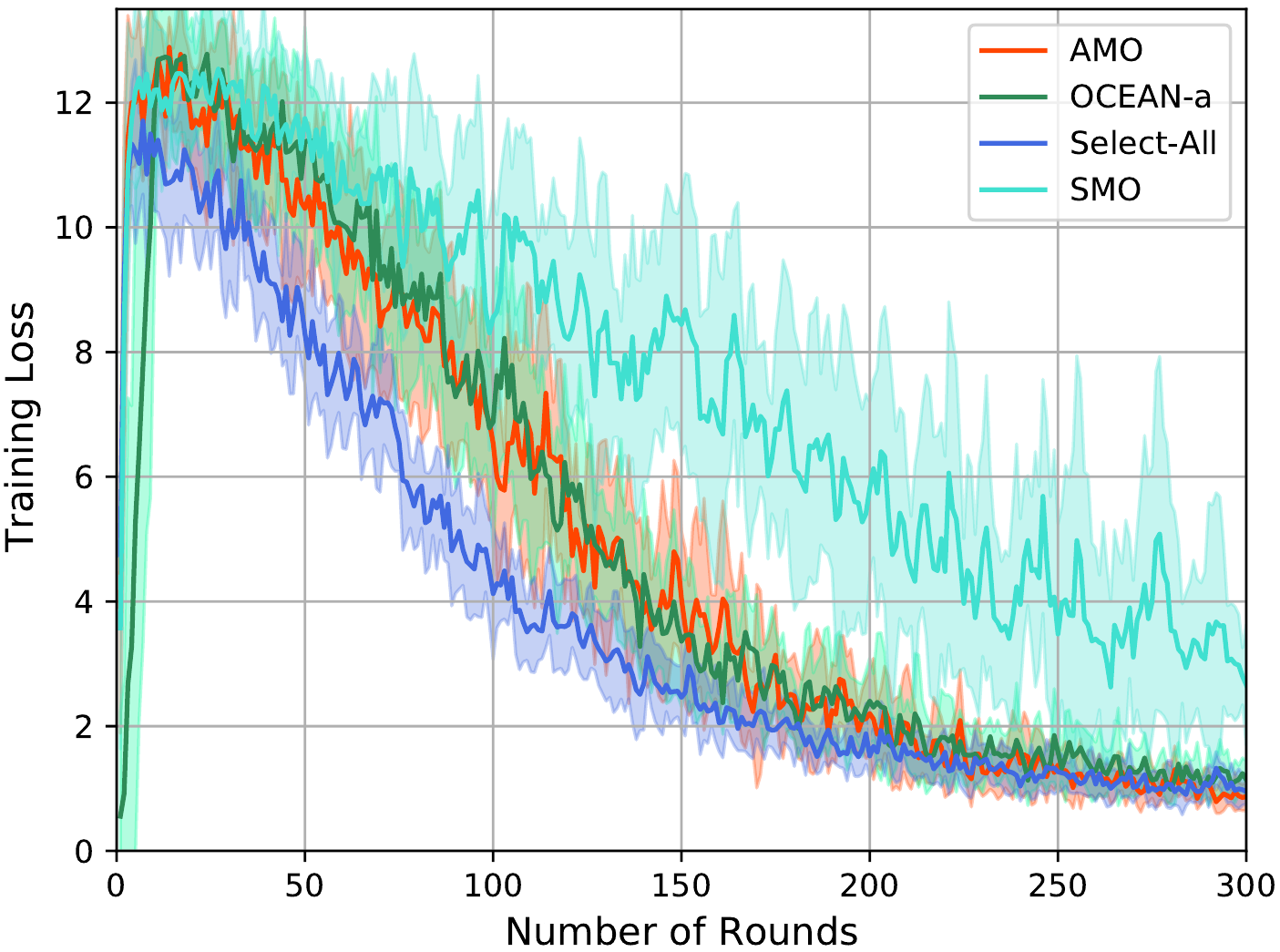}
\setlength{\abovecaptionskip}{0pt}
\caption{\label{fig5_1}Training Loss of OCEAN-a and Benchmarks}
\end{minipage}
\begin{minipage}[t]{0.48\textwidth}
\centering
\includegraphics[width=8cm]{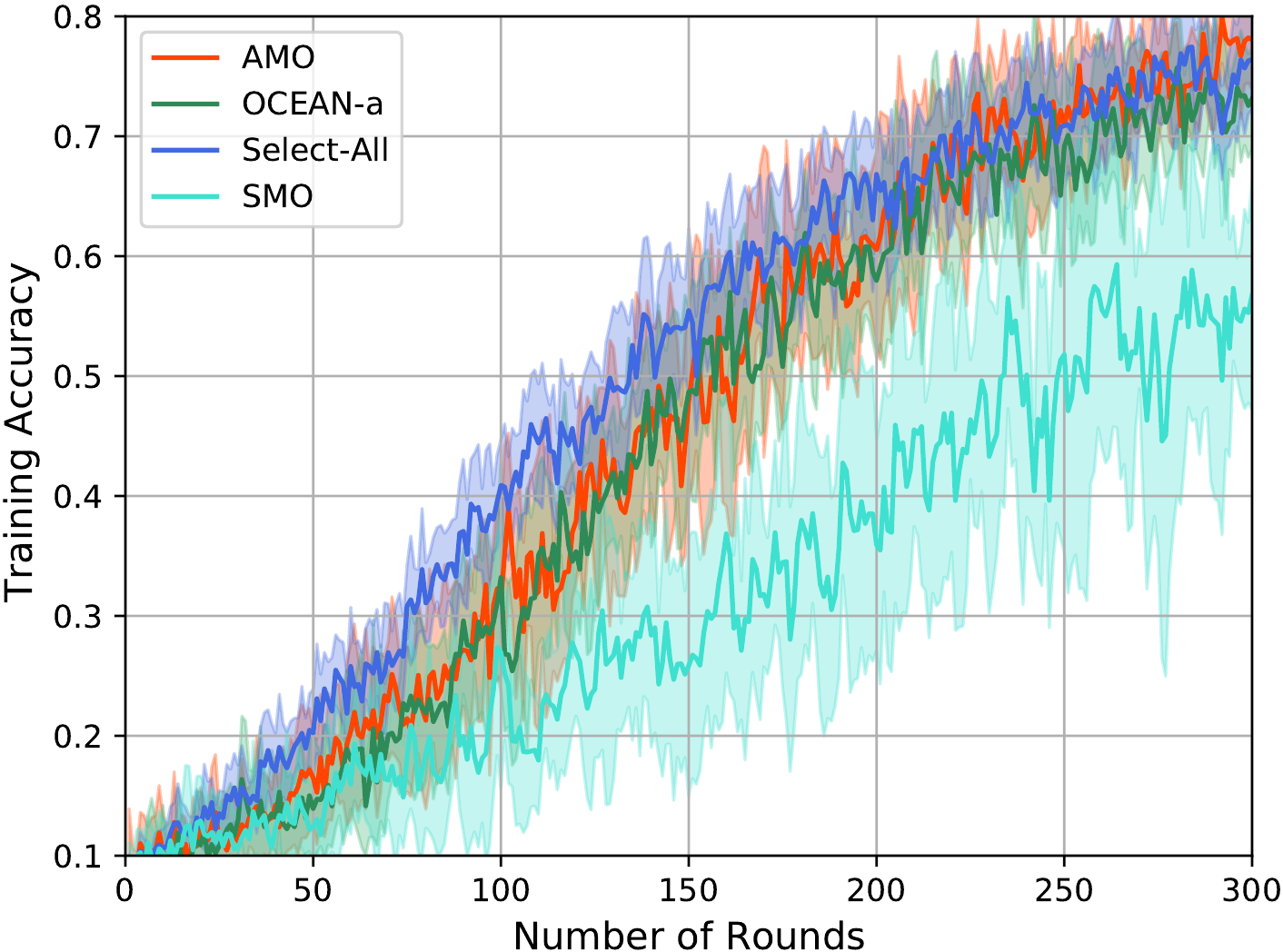}
\setlength{\abovecaptionskip}{0pt}
\caption{\label{fig5_2}Accuracy of OCEAN-a and Benchmarks}
\end{minipage}
\end{figure}

\subsection{Adaptability to Varying Network Condition}
Although the performance of AMO seems comparable to OCEAN-a in the last experiment, it is achieved in a relatively easy network, where the wireless channel is relatively stable. In this set of experiments, we simulate more challenging network environments where the wireless channel can vary considerably due to, e.g., client mobility. In particular, we simulate two scenarios. In Scenario 1, the average path loss gradually increases from 32 dB to 45 dB, mimicking a scenario where clients move away from the server over time. In Scenario 2, the average path loss gradually decreases from 45 dB to 32 dB, mimicking a scenario where clients move towards the server over time.

\underline{Scenario 1}. Figure \ref{fig6} shows the number of selected clients over 300 rounds for OCEAN-a and AMO and Figure \ref{fig7} shows the their FL accuracy. In the early rounds when the wireless channel is good, AMO selects some clients. However, as the channel gain degrades, AMO is not able to adapt to this change as the pre-allocated energy budget (even if the unused budget from the previous rounds is incorporated) cannot support even a single client to finish uploading the local model before the deadline $\bar{\tau}$. Only in the rounds towards the very end does the energy budget become sufficient and hence, some clients again are selected to upload their local model updates. Because of the long idle period in the middle when no clients are selected, the learning performance of AMO is significantly worse than OCEAN.

\underline{Scenario 2}.  Figure \ref{fig8} shows the number of selected clients over 300 rounds for OCEAN-a and AMO and Figure \ref{fig9} shows the their federated learning accuracy. In this scenario, the channel state in the early rounds is bad and hence, hardly any client can be selected to upload its local model update due to insufficient energy budget in AMO. As the channel state improves, AMO starts to select some clients but it becomes too late to do so.

In both scenarios, OCEAN is able to adapt its client selection decision because of its soft per-round energy budget allocation, yet the total consumed energy is still made close to the total energy budget. The per-client total energy consumption of OCEAN-a is shown in Figure \ref{fig10} for the two considered scenarios.

\begin{figure}[htbp]
\centering
\begin{minipage}[t]{0.48\textwidth}
\centering
\includegraphics[width=8cm]{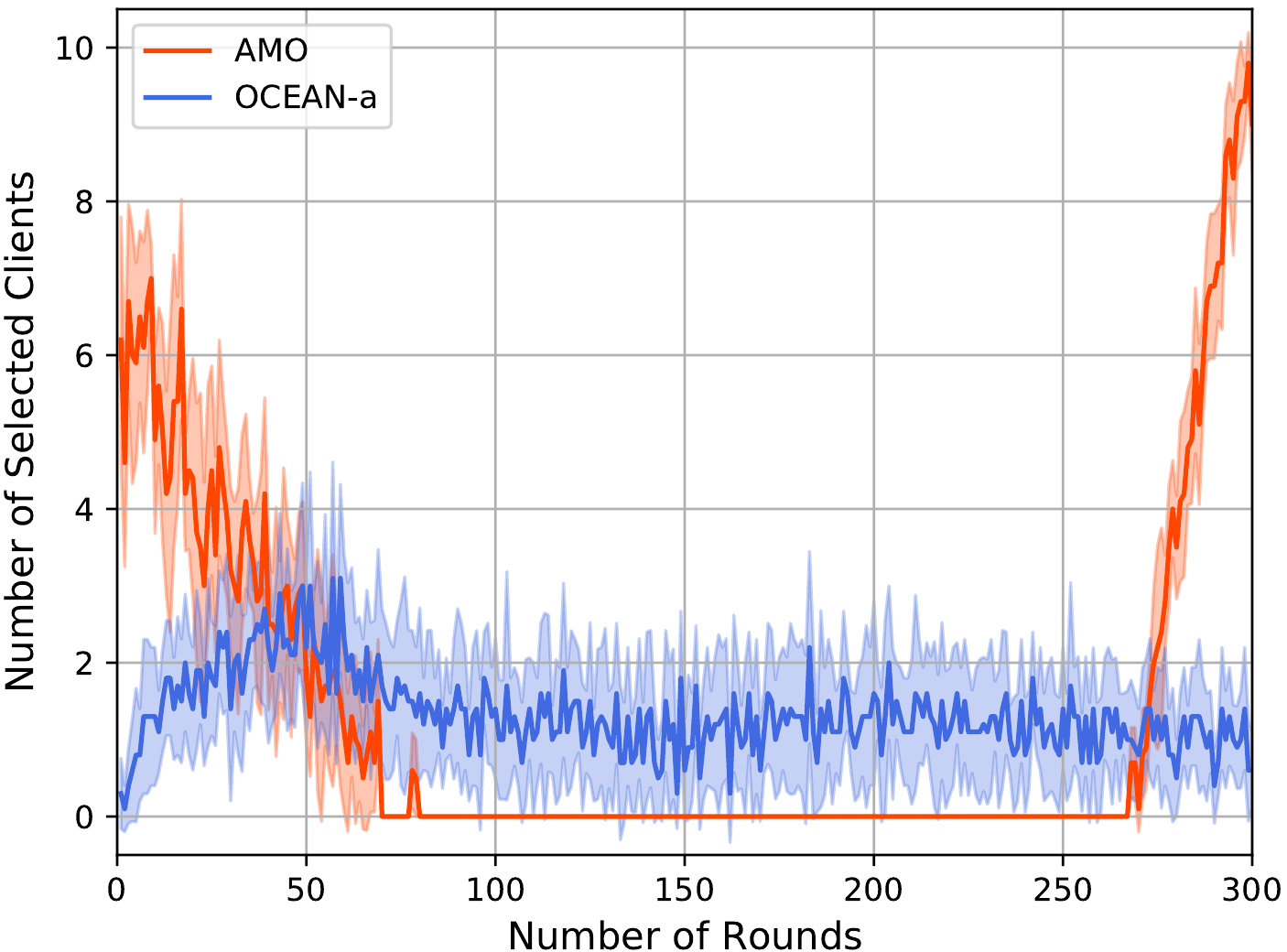}
\setlength{\abovecaptionskip}{0pt}
\caption{\label{fig6}Client Selection of Scenario 1}
\end{minipage}
\begin{minipage}[t]{0.48\textwidth}
\centering
\includegraphics[width=8cm]{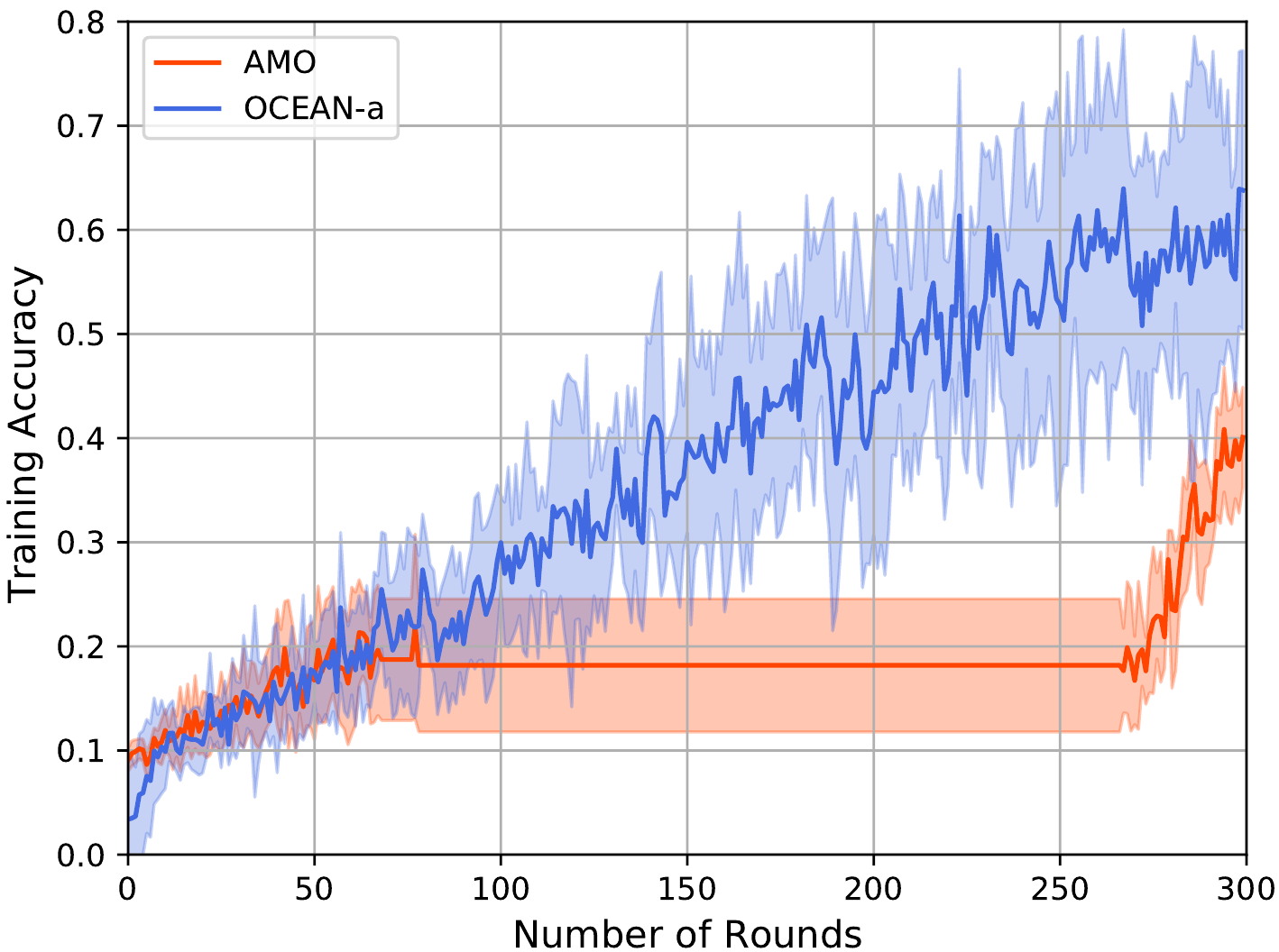}
\setlength{\abovecaptionskip}{0pt}
\caption{\label{fig7}Accuracy of Scenario 1}
\end{minipage}
\end{figure}

\begin{figure}[htbp]
\centering
\begin{minipage}[t]{0.48\textwidth}
\centering
\includegraphics[width=8cm]{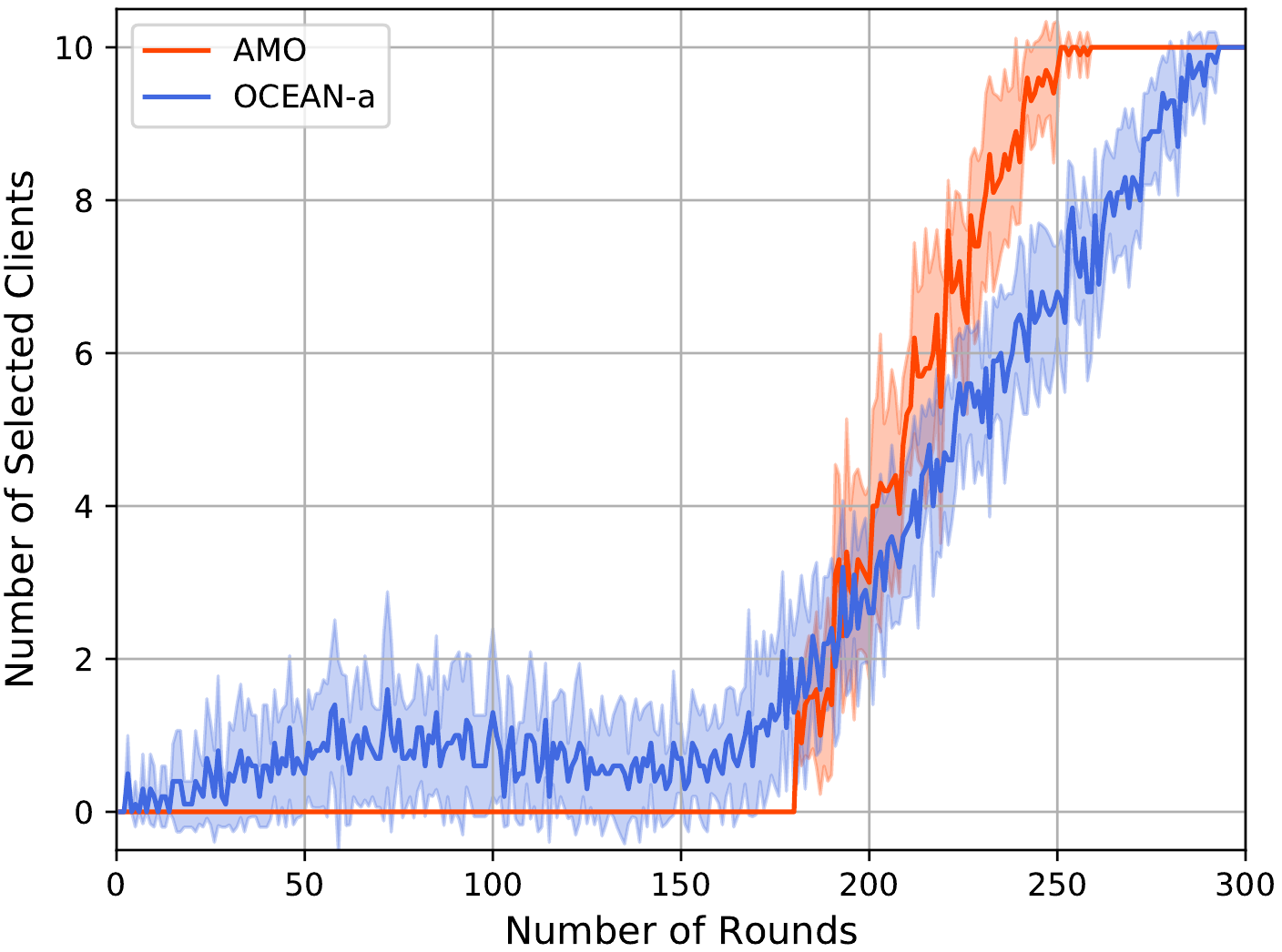}
\setlength{\abovecaptionskip}{0pt}
\caption{\label{fig8}Client Selection of Scenario 2}
\end{minipage}
\begin{minipage}[t]{0.48\textwidth}
\centering
\includegraphics[width=8cm]{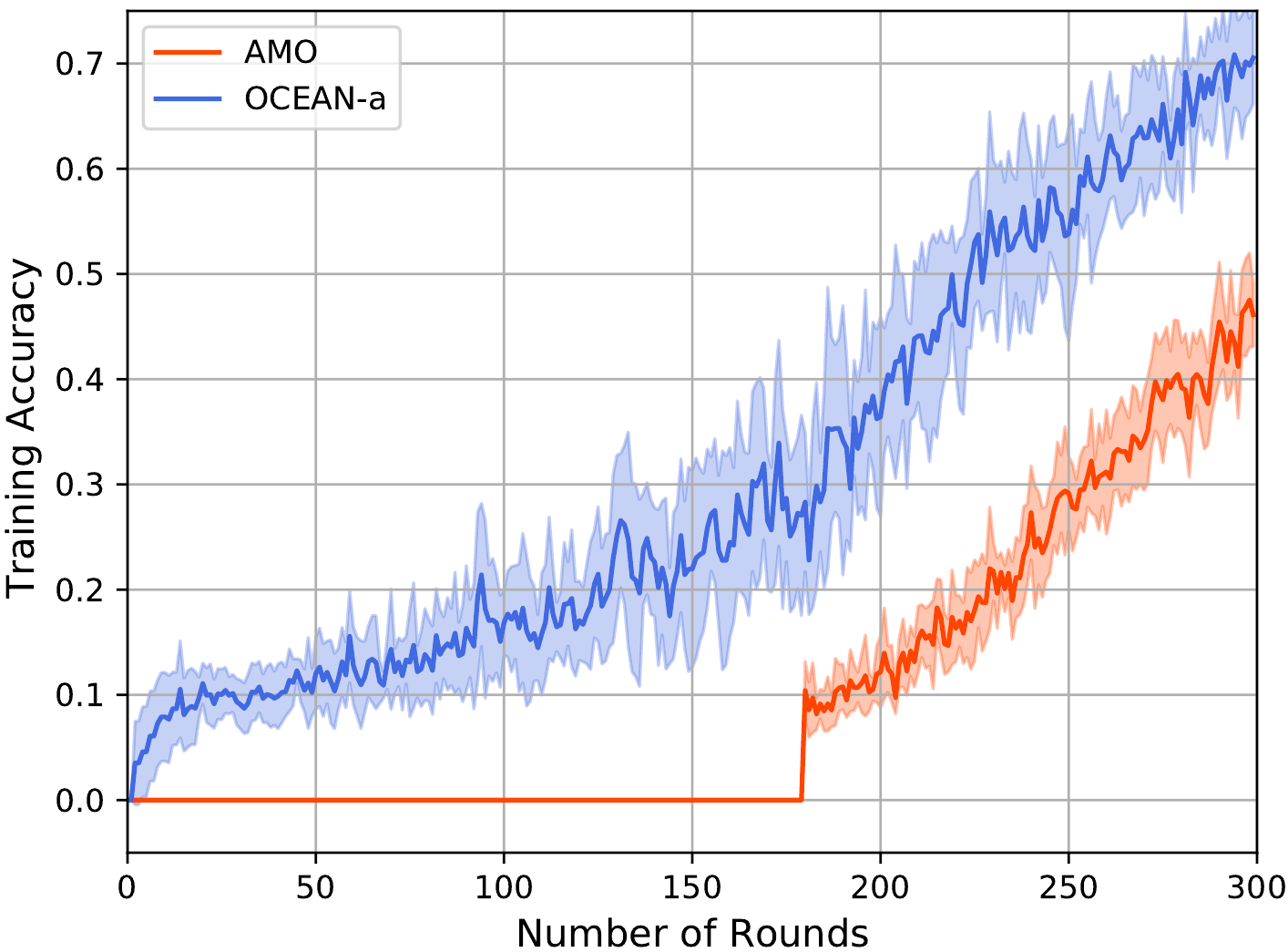}
\setlength{\abovecaptionskip}{0pt}
\caption{\label{fig9}Accuracy of Scenario 2}
\end{minipage}
\end{figure}

\begin{figure}[htbp]
\centering
\includegraphics[width=8cm]{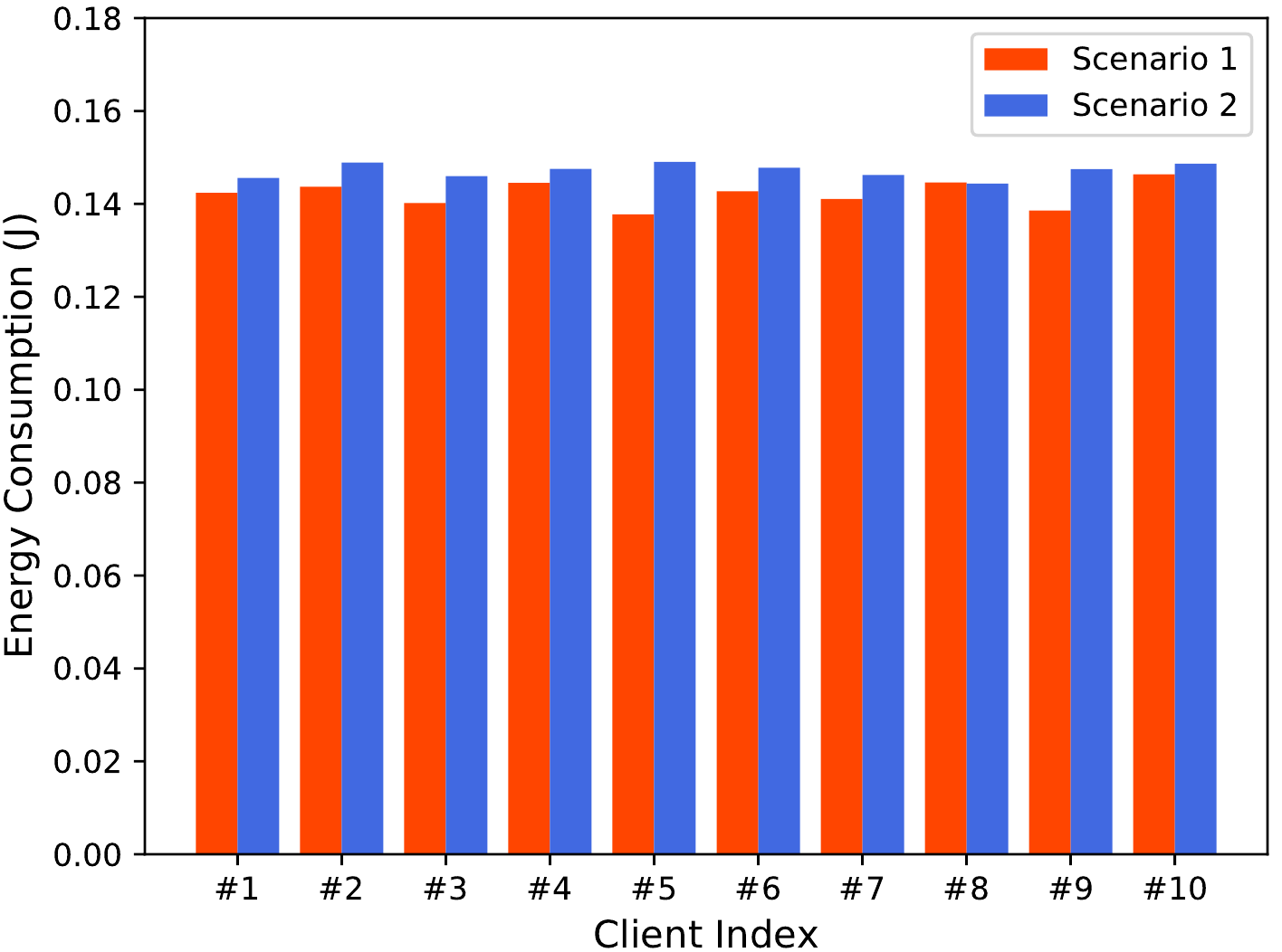}
\caption{\label{fig10}Energy Consumption of OCEAN-a for the Two Scenarios}
\end{figure}

\begin{figure}[htbp]
\centering
\begin{minipage}[t]{0.48\textwidth}
\centering
\includegraphics[width=8cm, height=5.7cm]{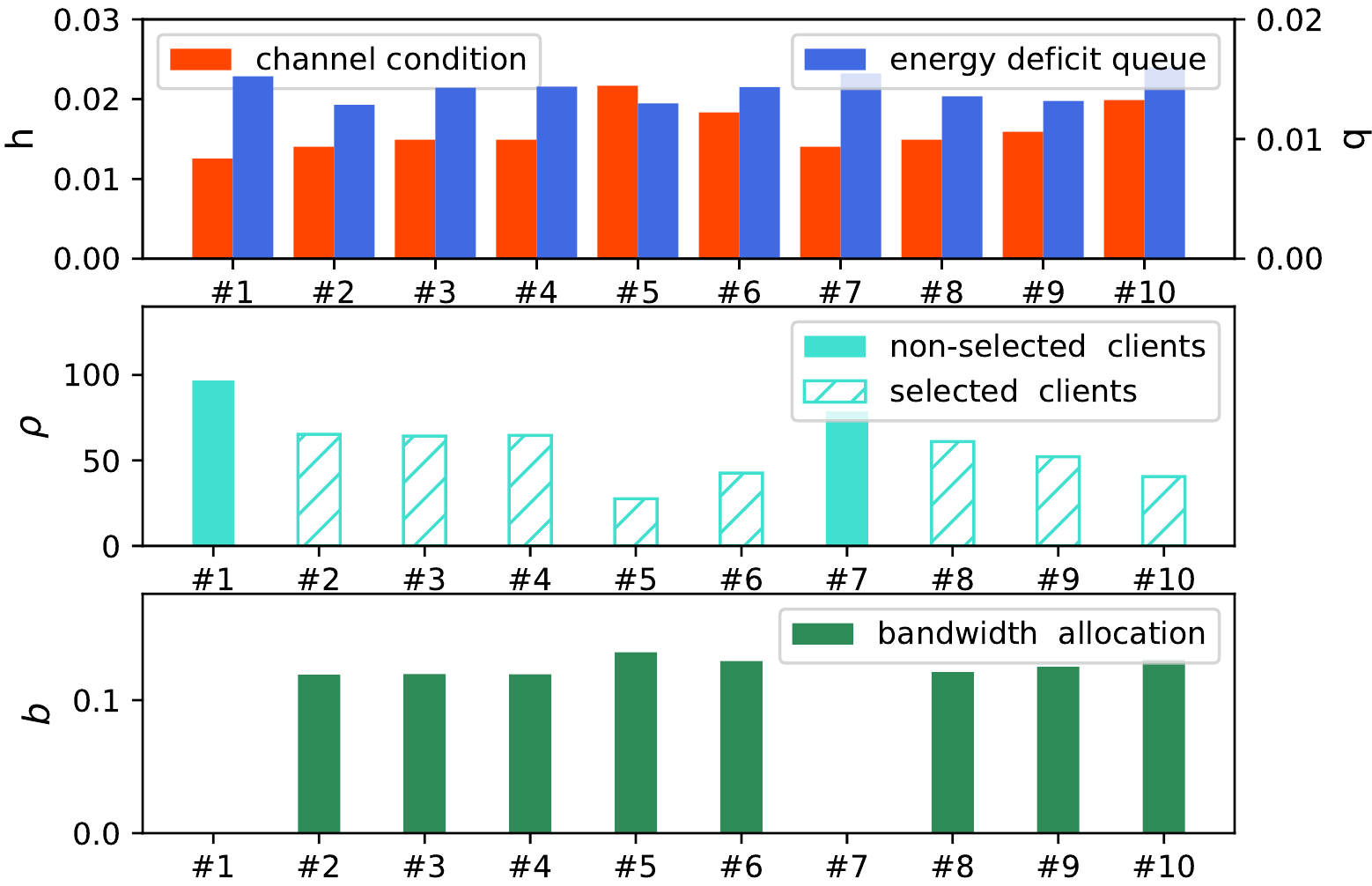}
\setlength{\abovecaptionskip}{0pt}
\caption{\label{fig11}Client Selection and Bandwidth Allocation Outcomes}
\end{minipage}
\begin{minipage}[t]{0.48\textwidth}
\centering
\includegraphics[width=8cm]{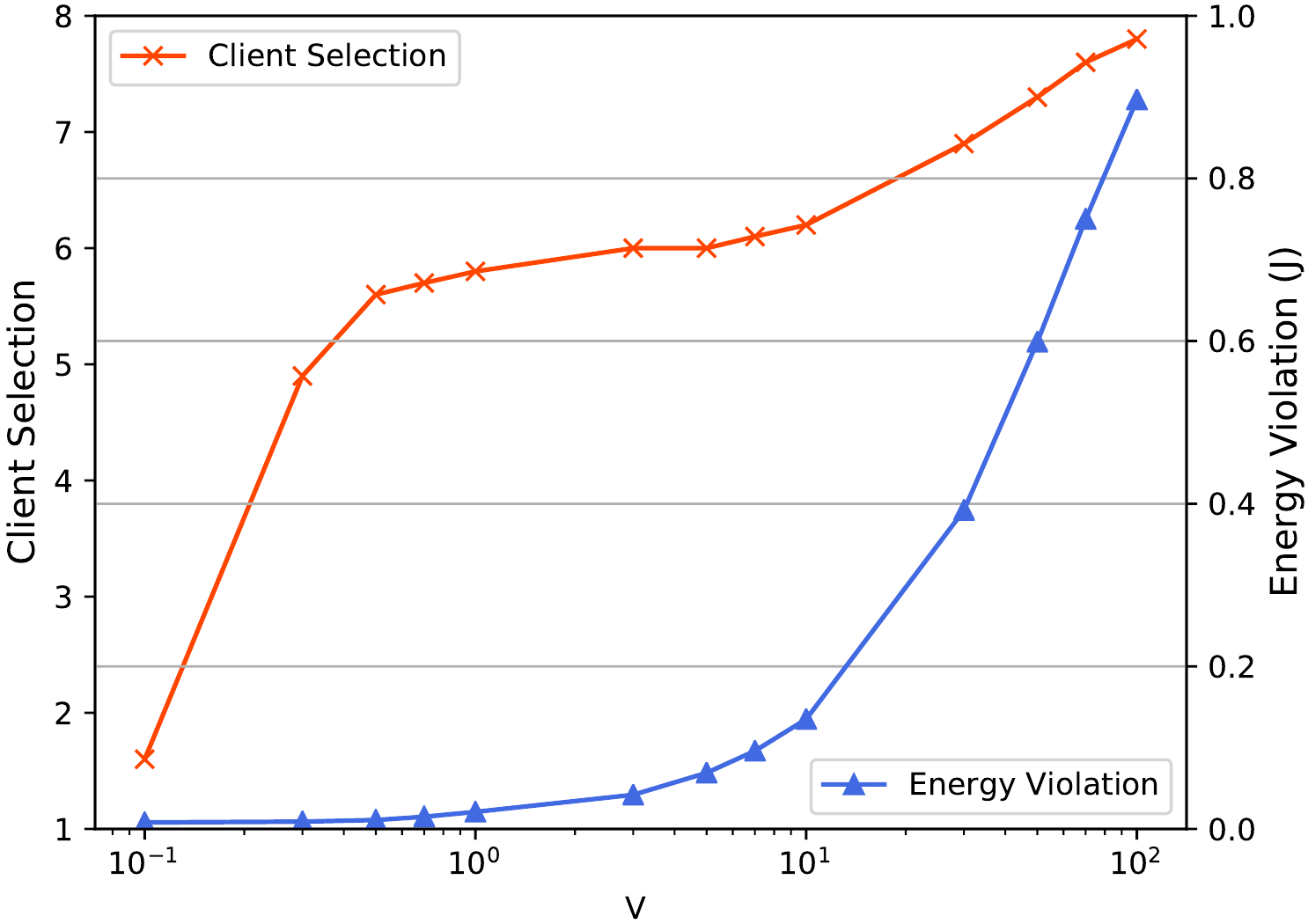}
\setlength{\abovecaptionskip}{0pt}
\caption{\label{fig12}Tradeoff Between Learning and Energy}
\end{minipage}
\end{figure}

\subsection{Features of OCEAN}
\subsubsection{Client Selection and Bandwidth Allocation Outcomes}
To have a deeper understanding of how OCEAN works, we illustrate, in one specific round, how clients are selected and bandwidth is allocated depending on the clients' channel condition and energy deficit queue in that round. In Figure \ref{fig11}, the upper subplot shows the current channel condition and energy deficit queue for each client. The middle subplot shows the computed selection priority $\rho$, with shaded bars indicating the selected clients. The bottom subplot shows the bandwidth allocation among the selected clients. As can be seen, a better channel condition and a larger deficit queue result in a higher priority (i.e., a smaller value of $\rho$). However, among the selected clients, more bandwidth is allocated to clients of a lower priority (i.e., a larger value of $\rho$).

\subsubsection{Learning - Energy Tradeoff}
Finally, we show the impact of the algorithm parameter $V$ on achieving different learning v.s. energy tradeoff of OCEAN. Figure \ref{fig12} shows the number of selected clients, the learning accuracy and the per-client energy consumption violation as a function of $V$. As can be seen, a larger $V$ emphasizes more on the learning performance, resulting in more selected clients and higher accuracy. On the other hand, a smaller $V$ emphasizes more on the energy consumption, resulting in a smaller violation (if any) on the total energy budget.

\section{Conclusion}
Resource allocation in wireless networks is an old topic, but it also faces constantly changing new challenges as new applications emerge. With FL being the trending new wireless network application, the old mindset of resource allocation for traditional applications such as file downloading or video streaming must be changed. This paper identifies a key property of FL, namely the temporal dependency and varying significance of learning rounds, that may significantly reshape how wireless resources should be allocated for optimized network and learning performance, yet is largely overlooked in the literature. While our formulation and algorithm have shown superior performance in real-world FL experiments, there are several future research directions that may extend the impact of this work. For example, we showed that an ascending client selection pattern is generally desired, but it is still not clear what the optimal pattern is. Moreover, client heterogeneity in terms of the computing power and local data size/distribution can be incorporated into the model to further enhance the understanding of resource allocation in more complex WFLNs.


%

\appendices
\section{Proof of Lemma 1}
\label{proofLemma1}
To prove the monotonicity and convexity of $f(x)$, we investigate the first and second order derivatives, respectively. The first-order derivative is
\begin{align}
f'(x) = 2^\frac{\beta}{x} - 1 + x\ln{2}\cdot 2^\frac{\beta}{x} (-\frac{\beta}{x^2})  = 2^\frac{\beta}{x}\left(1 - \ln{2}\cdot\frac{\beta}{x}\right)  -1
\end{align}

The second-order derivative is
\begin{align}
f''(x) = \ln{2}\cdot 2^\frac{\beta}{x} (-\frac{\beta}{x^2}) (1 - \ln{2}\cdot\frac{\beta}{x}) + 2^\frac{\beta}{x}\ln{2}\cdot \frac{\beta}{x^2} = (\ln{2})^2 2^\frac{\beta}{x}\frac{\beta^2}{x^3}
\end{align}
Therefore, for $x\in (0, \infty)$, $f''(x) > 0$ and hence $f'(x)$ is an increasing function. We also know that $\lim_{x\to \infty} f'(x) = 0$ and thus, $f'(x) < 0$ for $x \in (0, \infty)$.  This proves that $f(x)$ is decreasing on $(0, \infty)$. Similarly, for $x \in (-\infty, 0)$, $f''(x) < 0$ and hence $f'(x)$ is an decreasing function. Because $\lim_{x\to -\infty} f'(x) = 0$, $f'(x) > 0$ for $x \in (-\infty, 0)$. This proves that $f(x)$ is increasing on $(-\infty, 0)$. Moreover, $f(x)$ is convex on $(0, \infty)$ as $f''(x) > 0$; it is concave on $(-\infty, 0)$ as $f''(x) < 0$.

\section{Proof of Theorem 1}
\label{proofTheorem1}
The key is to prove that the optimal solution must have the following thresholding structure: there exists $k^*$ so that $a^*_k = 1, \forall k \leq k^*$ and $a^*_k = 0, \forall k > k^*$. To prove this, suppose that in the optimal solution  $(\a^*,\b^*)$, there exist $k_1 < k_2$ so that $a^*_{k_1} = 0$ and $a^*_{k_2} = 1$. Let us consider a different solution $(\tilde{\a}, \tilde{\b})$ which is obtained by swapping the decisions for $k_1$ and $k_2$ in $(\a^*,\b^*)$. Specifically,
\begin{align}
\tilde{a}_{k_1} = a^*_{k_2} = 1, \tilde{b}_{k_1} = b^*_{k_2}\\
\tilde{a}_{k_2} = a^*_{k_1} = 0, \tilde{b}_{k_2} = b^*_{k_1}
\end{align}
Since the decisions for other clients remain the same, the difference in the objective function value is
\begin{align}
&\sum_{k = 1}^K (\eta^t - \rho_k N_0\tilde{\tau}B f(\tilde{b}_k))\tilde{a}_k - \sum_{k = 1}^K (\eta - \rho_k N_0\tilde{\tau}B f(b^*_k)) a^*_k\\
=&(\eta-\rho_{k_1} N_0\tilde{\tau}B f(\tilde{b}_{k_1})) - (\eta-\rho_{k_2} N_0\tilde{\tau}B  f(b^*_{k_2})) = (\rho_{k_2} - \rho_{k_1}) N_0\tilde{\tau}B f(b^*_{k_2})) > 0
\end{align}
where $f(b) = b (2^\frac{L}{\tilde{\tau}B b} - 1)$. This is a contradiction to the optimality of $(\a^*, \b^*)$. Therefore, the optimal solution must have the aforementioned thresholding structure.

Next, we prove that the termination condition is correct. Let $i^*$ be the first client with $\eta - \rho_{i}N_0\tilde{\tau}Bf(b^*_{i^*}[i^*]) < 0$ where we use $b^*_k[i]$ to denote the optimal bandwidth allocation for the selection set $S = \{1, ..., i\}$. Clearly, when only clients $k \in \{1, ..., i^*-1\}$ are selected, we obtain a higher utility because
\begin{align}
&\sum_{k=1}^{i^*-1} (\eta - \rho_k N_0\tilde{\tau}Bf(b^*_k[i^*-1])) - \sum_{k=1}^{i^*} (\eta - \rho_k N_0\tilde{\tau}Bf(b^*_k[i^*])) \\
= &\sum_{k=1}^{i^*-1} (\eta - \rho_k N_0\tilde{\tau}Bf(b^*_k[i^*-1])) - \sum_{k=1}^{i^*-1} (\eta - \rho_k N_0\tilde{\tau}Bf(b^*_k[i^*]))\\
 &- (\eta - \rho_{i}N_0\tilde{\tau}Bf(b^*_{i^*}[i^*])) > 0
\end{align}
Therefore, client $i^*$ must not be in the optimal selection set. Because of the thresholding structure, we know that clients $k = i^* + 1, ..., K$ also must not be in the optimal selection set. This proves the correctness of the termination condition.

\section{Proof of Proposition 1}
\label{proofProposition1}
It suffices to consider the following optimization problem with two clients
\begin{align}
\min_{b_1, b_2}&~~ \rho_1 f(b_1) + \rho_2 f(b_2)\\
\text{s.t.}&~~b_1 + b_2 = \delta, ~~~b_1, b_2 \geq b_{min}
\end{align}
where $\delta$ is any constant in $(2 b_{\min}, 1]$. Let $\rho_1 < \rho_2$. Suppose the optimal bandwidth allocation satisfies $b^*_1 > b^*_2$, then by Lemma 1, we know $f(b^*_1) < f(b^*_2)$. Let us construct a different bandwidth allocation solution $\tilde{\b}$ where $\tilde{b}_1 = b^*_2$ and $\tilde{b}_2 = b^*_1$. In other words, the bandwidth allocation decisions are swapped. This solution also satisfies all constraints. We compare the respective objective values and have
\begin{align}
\rho_1 f(b^*_1) + \rho_2 f(b^*_2) - (\rho_1 f(\tilde{b}_1) + \rho_2 f(\tilde{b}_2)) =(\rho_1 - \rho_2)(f(b^*_1) - f(b^*_2)) > 0
\end{align}
This contradicts the optimality of $\b^*$. Therefore, we must have $b^*_1 \leq b^*_2$.

To prove $\rho_1 f(b^*_1) \leq \rho_2 f(b^*_2)$, let $b_2 = \delta - b_1$ and ignore the constraint that $b_1, b_2 \geq b_{min}$ for now. The first-order condition requires
\begin{align}
\rho_1 d f(b_1)/d b_1 + \rho_2 df(b_2)/d b_2 \cdot d b_2/ d b_1 = 0
\end{align}
This leads to
\begin{align}
\rho_1 f'(b^*_1) = \rho_2 f'(b^*_2)
\end{align}
Because $f'(x) < 0$, we can instead prove $f(b^*_1)/f'(b^*_1) \geq f(b^*_2)/f'(b^*_2)$. Let us define $g_1(x) \triangleq f(x)/f'(x)$. Since we have proven $b^*_1 \leq b^*_2$ in the above, we only need to prove that $g_1(x)$ is a non-increasing function in $x > 0$. To this end, consider the first order derivative of $g_1(x)$,
\begin{align}
g'_1(x) = \frac{(f'(x))^2 - f(x)f''(x)}{(f'(x))^2}
\end{align}
Let $g_2(x) \triangleq (f'(x))^2 - f(x)f''(x)$. We have to prove $g_2(x) \leq 0$ for $x > 0$.
\begin{align}
g_2(x)& = \left(2^\frac{\beta}{x}\left(1 - \ln{2}\cdot\frac{\beta}{x}\right)  -1\right)^2 - x(2^\frac{\beta}{x} -1)\cdot (\ln{2})^2 2^\frac{\beta}{x}\frac{\beta^2}{x^3}\\
& = (2^\frac{\beta}{x} - 1)^2 - 2(2^\frac{\beta}{x}-1)2^\frac{\beta}{x}\ln{2}\cdot\frac{\beta}{x} + (\ln{2})^2 2^\frac{\beta}{x}\frac{\beta^2}{x^2}
\end{align}
To simplify notations, we use a change of variable by letting $y = \beta/x$. Then proving $g_2(y) \leq 0$ for $y > 0$ is equivalent to proving $g_2(x) \leq 0$ for $x > 0$. We rewrite $g_2(y)$ below:
\begin{align}
g_2(y) \triangleq (2^y - 1)^2 - 2(2^y-1)2^y\ln{2}\cdot y + (\ln{2})^2 2^y y^2
\end{align}
Clearly, $g_2(0) = 0$. In order to prove $g_2(y) \leq 0$, we prove $g_2(y)$ is decreasing in $y \geq 0$.
\begin{align}
g'_2(y) = -(\ln{2})^2 y 2^{y}\left(4(2^y - 1) - \ln{2}\cdot y\right)
\end{align}
It is easy to verify that $g_3(y) \triangleq 4(2^y - 1) - \ln{2}\cdot y$ is increasing in $y > 0$ and $g_3(0) = 0$, which means $g_3(y) \geq 0$ for $y \geq 0$. Hence, $g'_2(y) < 0$ for $y > 0$. This concludes the proof for $\rho_1 f(b^*_1) \leq \rho_2 f(b^*_2)$ by ignoring the constraint $b_1, b_2 \geq b_{min}$.

When the constraint $b_1, b_2 \geq b_{min}$ is considered, there are two cases. In the first case, $b_{min} \leq b^*_1 \leq b^*_2$. In this case, the constraint is automatically satisfied and hence, our above conclusion holds. In the second case, $b^*_1 \leq b_{min} \leq b^*_2$. In this case, the optimal allocation is modified to $\tilde{b}^*_1 = b_{min} \geq b^*_1$ and $\tilde{b}^*_2 = 1-b_{min} \leq b^*_2$. Since $f(x)$ is a decreasing function, $\rho_1 f(\tilde{b}^*_1) \leq \rho_1 f(b^*_1) \leq \rho_2 f(b^*_2) \leq \rho_2 f(\tilde{b}^*_2)$. This completes the proof.

\section{Proof of Theorem 2}
\label{proofTheorem2}
We define the quadratic Lyapunov function $L(\q(t)) \triangleq \frac{1}{2}\sum_{k=1}^K q^2_k(t)$. Let $\Delta_1(t)$ be the 1-round Lyapunov drift yielded by some control decisions over one round: $\Delta_1(t) \triangleq L(\q(t+1)) - L(\q(t))$. Similarly, let $\Delta_R(t)$ be the $R$-round Lyapunov drift: $\Delta_R(t) \triangleq L(\q(t+R)) - L(\q(t))$. Based on the queue dynamics, we have
\begin{align}
\frac{1}{2}\sum_{k}^K q^2_k(t+1) &\leq \frac{1}{2}\sum_{k=1}^K [E_k(a^t_k, b^t_k|h^t_k) - H_k/T + q_k(t)]^2
\end{align}
Then, it can be easily show that
\begin{align}
\Delta_1(t) \leq C_1 +\sum_{k=1}^K q_k(t)\cdot[E_k(a^t_k, b^t_k|h^t_k) - H_k/T]
\end{align}
where $C_1$ is a constant satisfying $C_1 \geq \frac{1}{2}\sum_{k=1}^K (E^{\text{max}} - H^\text{min}/T)^2, \forall t$, which is finite due to the boundedness of the channel condition $h^t_k$ and the minimum bandwidth allocation requirement $b_{min}$. Next, it is straightforward that $\forall m$ and  $\forall t = mR, ..., (m+1)R - 1$
\begin{align}\label{one-slot-dpp}
V_m\cdot U(\a^t) - \Delta_1(t) \geq V_m\cdot U(\a^t) - \sum_{k=1}^K q_k(t)\cdot[E_k(a^t_k, b^t_k|h^t_k) - H_k] - C_1
\end{align}
As we can see, by solving \textbf{P3}, OCEAN-P actually maximizes a lower bound of $V_m\cdot U(\a^t) - \Delta_1(t)$. Let $\hat{\a}^0, \hat{\b}^0$, $...$, $\hat{\a}^{T-1}, \hat{\b}^{T-1}$ be the sequence of decisions derived by the online algorithm.

(a) Consider a specific sequence of decisions where $\tilde{a}^t_k = 0, \forall t, k$. Clearly, in this case, $U(\tilde{\a}^t) = 0$ and $E_k(\tilde{a}^t_k, \tilde{b}^t_k|h^t_k) - H_k/K  = -H_k/K$. Because $\hat{a}^t, \hat{b}$ maximizes the right-hand side of \eqref{one-slot-dpp}, we have
\begin{align}
V_m\cdot U(\hat{\a}^t) - \Delta_1(t) \geq  V_m\cdot 0 +\sum_{k=1}^K \tilde{q}_k(t) H_k - C_1 \geq -C_1
\end{align}
Therefore
\begin{align}
\Delta_1(t) \leq  V_m\cdot U(\hat{\a}^t) + C_1 \leq V_m \eta^t K + C_1
\end{align}
As enforced by the online algorithm,
\begin{align}
\Delta_R(mR) = \frac{1}{2}\sum_{k=1}^K (q^2_k(mR + R) - q^2_k(mR)) = \frac{1}{2}\sum_{k=1}^K q^2_k(mR + R)
\end{align}
is the $R$-round drift calculated after the $m$-th rest but before the $(m+1)$-th reset of the energy deficit queue (so $q_k(mR) = 0$). Thus, before the $(m+1)$-th reset of the energy deficit queue, we have
\begin{align}
\sum_{k=1}^K q^2_k(mR + R) = 2\Delta_R(mR) =2 \sum_{t=mR}^{mR+R-1}\Delta_1(t) \leq 2R(V_m\eta^t K + C_1)
\end{align}
Therefore, $\forall k$,
\begin{align}\label{qbound}
q_k(mR + R) \leq \sqrt{2R(V_m\eta^t K + C_1)}
\end{align}

On the other hand, according to the queue dynamics \eqref{queue}, we have
\begin{align}
q_k(t+1) - q_k(t) \geq E_k(a^t_k, b^t_k|h^t_k) - H_k/T
\end{align}
Summing both sides over the rounds in the $m$-th frame, namely $t = mR, ..., (m+1)R - 1$, and dividing by $R$, we have
\begin{align}\label{energy_deficit}
\frac{1}{R}\sum_{t = mR}^{(m+1)R - 1}(E_k(a^t_k, b^t_k|h^t_k) - H_k/T) \leq \frac{q_k((m+1)R) - q_k(mR)}{R} = \frac{q_k((m+1)R)}{R}
\end{align}

Plugging \eqref{qbound} into \eqref{energy_deficit}, we have
\begin{align}
\frac{1}{R}\sum_{t = mR}^{(m+1)R - 1}(E_k(\hat{a}^t_k, \hat{b}^t_k|h^t_k) - H_k/T) \leq \sqrt{\frac{2(V_m\eta^t K + C_1)}{R}}
\end{align}
Considering all $M$ frames, we obtain
\begin{align}
\sum_{t=0}^T E_k(\hat{a}^t_k, \hat{b}^t_k|h^t_k) \leq H_k + \sum_{m=0}^{M-1} \sqrt{\frac{2(V_m\eta^t K + C_1)}{R}}, \forall k
\end{align}

(b) Consider the $R$-round weighted learning utility minus drift:
\begin{align}
&V_m\sum_{t=mR}^{mR+R-1}U(\a^t) - \Delta_R(mR) \\
\geq & V_m\sum_{t=mR}^{mR+R-1}U(\a^t) - C_1R - \sum_{t=mR}^{mR+R-1}\sum_{k=1}^K  q_k(t)\cdot[E_k(a^t_k, b^t_k|h^t_k) - H_k/T]
\end{align}
Because $\hat{\a}^0, \hat{\b}^0, ..., \hat{\a}^{T-1},\hat{\b}^{T-1}$ explicitly maximizes the right-hand side of the above equation, the following must also hold
\begin{align}
&V_m\sum_{t=mR}^{mR+R-1}U(\hat{\a}^t) - \Delta_R(mR)\\
\geq &V_m\sum_{t=mR}^{mR+R-1}U(\a^{*,t}) - C_1R - \sum_{t=mR}^{mR+R-1}\sum_{k=1}^K  q_k(t)\cdot[E_k(a^{*,t}_k, b^{*,t}_k|h^t_k) - H_k/T]\\
\geq &V_m\sum_{t=mR}^{mR+R-1}U(\a^{*,t}) - C_1 R - \sum_{t=mR}^{mR+R-1}\sum_{k=1}^K (t-mR) E^{\text{max}} \cdot[E_k(a^{*,t}_k, b^{*,t}_k|h^t_k) - H_k/T] \\
&- \sum_{k=1}^K q_k(mR)\sum_{t=mR}^{mR+R-1} [E_k(a^{*,t}_k, b^{*,t}_k|h^t_k) - H_k/T]\\
\geq & V_m\sum_{t=mR}^{mR+R-1}U(\a^{*,t}) - \sum_{k=1}^K q_k(mR)\sum_{t=mR}^{mR+R-1} [E_k(a^{*,t}_k, b^{*,t}_k|h^t_k) - H_k/T] - \left(C_1R + \frac{R(R-1)K}{2} (E^{\text{max}})^2\right)\\
\geq & V_m\sum_{t=mR}^{mR+R-1}U(\a^{*,t}) - \left(C_1R + \frac{R(R-1)K}{2} (E^{\text{max}})^2\right)
\end{align}
where in $\hat{\a}^{*,0}, \hat{\b}^{*,0}, ..., \hat{\a}^{*,T-1},\hat{\b}^{*,T-1}$ is the optimal decision that solves the $R$-round lookahead problems. Notice that $q_k(t)$ in the above equation is still derived by OCEAN-P. The first inequality holds because OCEAN-P maximizes the lower bound. The last inequality holds because $q_k(mR)$ is reset to zero as enforced by OCEAN-P.

Noticing $\Delta_R(mR) \geq 0$ and dividing both sides by $V_m$, we have
\begin{align}
\sum_{t=mR}^{mR+R-1}U(\hat{\a}^t) \geq  \sum_{t=mR}^{mR+R-1}U(\a^{*,t}) - \frac{1}{V_m}\left(C_1R + \frac{R(R-1)K}{2} (E^{\text{max}})^2\right)
\end{align}
By summing over $m = 0, ..., M-1$, we have
\begin{align}
\sum_{t=0}^{T-1}U(\hat{\a}^t) \geq \sum_{m=0}^{M-1}U_m^* - C_2 \sum_{m=0}^{M-1}\frac{1}{V_m}
\end{align}
where $C_2 \triangleq C_1R + \frac{R(R-1)K}{2} (E^{\text{max}})^2$.


\bibliographystyle{IEEEtran}
\bibliography{bibligraphy}

\end{document}

%% file: macros.tex
\newcommand{\bp}{\begin{proof} \small }
\newcommand{\ep}{\end{proof} \normalsize}
\newcommand{\epx}{\end{proof} \small}
\newcommand{\bpa}{\begin{proofappx} \footnotesize }
\newcommand{\epa}{\end{proofappx} \small }
\newtheorem{theorem}{Theorem}
\newtheorem{proposition}{Proposition}

\newtheorem{lemma}{Lemma}

\newtheorem*{theorem*}{Theorem}
\newtheorem*{proposition*}{Proposition}
\newtheorem*{corollary*}{Corollary}
\newtheorem*{lemma*}{Lemma}
\newtheorem*{assumption*}{Assumption}
\newtheorem*{definition*}{Definition}
\newtheorem*{claim*}{Claim}

\newcommand{\be}{\begin{equation}}
\newcommand{\ee}{\end{equation}}
\newcommand{\bs}{\begin{subequations}}
\newcommand{\es}{\end{subequations}}
\newcommand{\bq}{\begin{eqnarray}}
\newcommand{\eq}{\end{eqnarray}}
\newcommand{\bqn}{\begin{eqnarray*}}
\newcommand{\eqn}{\end{eqnarray*}}

\newcommand{\ba}{\left[ \begin{array}}
\newcommand{\ea}{\\ \end{array} \right]}
\newcommand{\ben}{\begin{enumerate}}
\newcommand{\een}{\end{enumerate}}

\def\a{{\boldsymbol{a}}}
\def\b{{\boldsymbol{b}}}

\def\q{{\boldsymbol{q}}}

\def\real{{\mathchoice%
{\hbox{\rm\setbox1=\hbox{I}\copy1\kern-.45\wd1 R}}
{\hbox{\rm\setbox1=\hbox{I}\copy1\kern-.45\wd1 R}}
{\hbox{\scriptsize\rm\setbox1=\hbox{I}\copy1\kern-.45\wd1 R}}
{\hbox{\scriptsize\rm\setbox1=\hbox{I}\copy1\kern-.45\wd1 R}}}}

\def\Zint{{\mathchoice{\setbox1=\hbox{\sf Z}\copy1\kern-.75\wd1\box1}
{\setbox1=\hbox{\sf Z}\copy1\kern-.75\wd1\box1}
{\setbox1=\hbox{\scriptsize\sf Z}\copy1\kern-.75\wd1\box1}
{\setbox1=\hbox{\scriptsize\sf Z}\copy1\kern-.75\wd1\box1}}}
\newcommand{\complex}{ \hbox{\rm C\kern-0.45em\rule[.07em]{.02em}{.58em}%
\kern 0.43em}}